\documentclass[11pt]{article}

\usepackage[margin=1in]{geometry}
\usepackage{authblk} 
\usepackage{amsmath,amssymb,amsthm,bm}
\usepackage{enumitem}
\usepackage[round]{natbib}
\usepackage{booktabs} 
\usepackage{subcaption}
\usepackage{graphicx} 
\usepackage{hyperref}
\hypersetup{colorlinks=true,citecolor=blue,linkcolor=blue,urlcolor=blue}

\newtheorem{theorem}{Theorem}
\newtheorem{definition}{Definition}

\theoremstyle{remark}
\newtheorem{remark}{Remark}
\theoremstyle{definition}
\newtheorem{assumption}{Assumption}

\newcommand{\E}{\mathbb{E}}

\newcommand{\dhat}{\widehat{\delta}}

\newcommand{\Var}{\operatorname{Var}}

\newcommand{\Kcal}{\mathcal{K}}

\usepackage[plain,noend]{algorithm2e}

\title{\bf De-biased Distribution Estimation for Locally Differential Private Survival Data \\
}

\author[1,*]{
Yi Xiong}
\author[2]{Li-Pang Chen}
\author[3]{Liyue Fan}
\affil[1]{Department of Biostatistics, The State University of New York (SUNY) at Buffalo}

\affil[2]{Department of Statistics, National Chengchi University}

\affil[3]{Department of Computer Science, University of North Carolina at Charlotte}
 \affil[*]{Corresponding author: \texttt{yxiong3@buffalo.edu}}
 \date{}
\begin{document}

\maketitle
\begin{abstract}
Sharing time-to-event data can leak sensitive information about individual participants in clinical studies. To protect the privacy of data, differential privacy (DP) framework offers rigorous protection for data sharing. The local variant of differential privacy is particularly appealing because each record is perturbed before it leaves the data owner or contributor. However, the injected perturbation distorts the true distribution of the source data, yielding low utility and invalid statistical inference with the released data. In this work, we propose a nonparametric estimator for the true survivor function via the deconvolution kernel function with locally differential private (LDP) right-censored survival data in which both the event times and the event indicators are privatized. Our estimator simultaneously de-biases the perturbation applied to the times and the errors induced by flipping the event indicators due to LDP. Theoretically, we establish the uniform consistency and point-wise asymptotic normality of the de-biased distribution estimator, and derive its uniform convergence rate together with the optimal bandwidth attaining that rate. Simulation studies and the analysis of a semi-synthetic dataset generated from the METABRIC breast cancer cohort demonstrate the performance of the proposed estimator across privacy budgets, censoring levels, and sample sizes. By the post-processing property of DP, the proposed de-biased estimator inherits the privacy guarantee while restoring valid inference on the survivor function.

\medskip
\noindent\textbf{Keywords:} censoring; deconvolution; differential privacy; kernel density estimation; randomized response; survival analysis.
\end{abstract}

\section{Introduction}
\label{sec:intro}


Time-to-event data, such as times to disease diagnosis, hospitalization, or
death, is among the most sensitive records held by registries, electronic
health systems, and cohort studies. Because releasing individual-level data on the exact event times is
risky, it is common to share the estimated survival curves, which is usually obtained by the Kaplan--Meier
estimator in registries, trial
reports, and published articles. Such a release, however, does not guarantee privacy protection. The Kaplan--Meier estimator is a step function whose jumps occur
exactly at the observed event times and are determined by the
number of subjects at risk, so a published curve, together with the at-risk
counts are likely to disclose individual event times and to expose participants to membership inference (e.g., \citealp[]{Bonomi2020}).

To protect individual data privacy while still enabling data sharing, \emph{differential
privacy} (DP, \citealp{dwork2006calibrating}) has become the standard for quantifiable protection. In its central mode, a trusted curator
holds the original records and releases only privatized summaries computed from
them. The guarantee then attaches to what is released. As the original individual-level data are stored in a centralized repository, security breaches may occur at the repository, inflicting data leak and thus loss of privacy. For example, in June 2026, an
experimental Artificial Intelligence (AI) agent with internet access was gathering public statistics on
medical cost from a public-facing portal operated by Australia's Medicare
system\citep{NatureMedicare2026}. Implementing protection in the local differential privacy mode provides stronger privacy protection in the presence of potential security attacks on such a central repository. In the local mode of differential privacy, each
individual's record is randomized before it leaves their control
\citep{kasiviswanathan2011what,duchi2018minimax}, so the data curator never sees
the truth. Local differential privacy (LDP) is attractive because it
makes a worst-case guarantee about every released record and this holds
regardless of what later happens to the repository. In the framework of survival analysis, LDP aims to protects survival data by adding perturbation to time-to-event outcomes or censoring status. This mathematically guarantees that an attacker cannot tell if the exact event time and event indicator of an inidividual was in the dataset. 

Existing work on privacy-preserving survival analysis has adopted both central and local DP models. \citet{Bonomi2020} and \citet{FanBonomi2023} added perturbation to
summary statistics, for instance the counts of events at given time points or
the gradients used to train a deep learning model, and \citet{Gondara2020}
obtained a differentially private Kaplan--Meier estimator of the survival
function together with confidence intervals. To protect the individual-level
data, \citet{bonomi2022sharing} presented a data-sharing mechanism that
randomizes each event time within a truncated time window. \cite{Maxime2025} considered an LDP mechanism that only privatizes the event indicators and releasing the observed times unperturbed,
under which a kernel estimator of the cumulative hazard remains consistent. However, the observed times are themselves sensitive and remain unprotected in their mechanism. As \citet{bonomi2022sharing} pointed out, randomizing event times only within a bounded window offers weaker protection than a local DP mechanism, which
perturbs each time-to-event record and randomizes each event indicator. On
the other hand, they also noted that the guarantee of local DP is achieved with randomized mechanisms, such as added perturbation, which may lower the utility of the released data and
distorts their true distribution.

Fortunately, this distortion is not arbitrary. Because the LDP perturbation is
generated by a published mechanism, its distribution is known exactly, so
privatized data are not merely noisier data but data contaminated by an error of
known law. The bias it induces can therefore be removed by an appropriate
post-processing step, at the cost of additional variance rather than validity.
Moreover, by the post-processing property of DP, such a correction consumes no
additional privacy budget. This observation motivates our work. Indeed, the released survival data can be viewed as error-contaminated versions
of the original data, which connects our problem to the measurement error
literature (e.g., \citealp{carroll2006measurement}). For time-to-event data,
correction methods in this literature have been developed predominantly for
regression analysis. For example, \citet{Oh2018} and \citet{Oh2021} extended
simulation and extrapolation (SIMEX) and regression calibration to address
time-to-event error in the parametric accelerated failure time (AFT) model, and
\citet{Chen2024} further considered simultaneous time-to-event error and flipped
censoring status, proposing a boosting method for variable selection and
estimation in the AFT model with nonlinear covariates. These corrections,
however, do not carry over to the survivor function and neither of SIMEX nor regression calibration can fully correct the biases that vanish as the sample size grows. Beyond regression, nonparametric estimation of the marginal survivor function
from error-contaminated time-to-event data, which is a common objective in
survival analysis, has received little attention. \citet{Zee2015} considered
discrete survival times and proposed a nonparametric likelihood method under the
structural measurement error model. \citet{Comte2017} combined a Fourier-based
deconvolution with B-spline regression to estimate the hazard function, and 
most recently \citet{Edwards2026} extended the Rogan--Gladen estimator to
correct an estimated survivor function for error-prone time-to-event variables. However, to our knowledge, no existing work addresses survival data in which  the event times carry measurement error and the indicators carry misclassification error simultaneously, which is the situation created by the local differential privatization on both the times and the indicators.

In this paper, we develop a nonparametric framework for estimating the survivor
function from right-censored data released under local differential privacy, in
which both the event times and the event indicators are privatized. Rather
than calibrating individual observations, our framework corrects the
perturbation-induced bias at the level of the distribution, using deconvolution
kernels for the times and moment reconstruction for the event indicators. Deconvolution kernel estimation dates back to \citet{stefanski1990deconvolving},
who recovered the density of a variable observed with additive error of known
distribution, and \citet{fan1991optimal} established its optimal rates of
convergence. More recently, \citet{farokhi2020deconvoluting} brought this
technique to locally differentially private data, using deconvolution kernels
for density estimation and regression. These methods, however, concern
fully observed data instead of censored data. Local
differential privacy differs fundamentally from conventional measurement error
in that the distortion is introduced deliberately, by a mechanism designed to
satisfy a privacy constraint. Consequently, the distribution of the perturbation is known and correcting the biases will not essentially require any validation data. To
account for censoring, we propose a Nelson--Aalen type estimator built
from deconvolution kernel estimators of the event sub-density and the at-risk function.
Because it works through the hazard, the estimator requires neither a model for
the survival times nor knowledge or estimation of the censoring distribution.
As the framework is built completely on the privatized data itself, it is a post-processing
of the released data and inherits their LDP guarantee
\citep[Proposition~2.1]{dwork2014algorithmic}. We derive uniform
convergence rate that depends explicitly on the privacy budget and establish the uniform consistency together with
the point-wise normality of the proposed estimator asymptotically. We also justify the validity of
a multiplier bootstrap procedure for variance estimation. Since the convergence rate depends on the sample size and the privacy budget, it makes the statistical cost of local privacy explicit. It can further guide the design of data-sharing protocols on how the sample size are required to scale with the privacy budget to attain that accuracy of proposed estimator.


The remainder is organized as follows. In Section~\ref{sec:setup}, we introduce the basic setup including the privatization mechanism using local differential privacy on survival data subject to right censoring, and we present the proposed distribution estimator with the LDP survival data and the asymptotic properties of the proposed estimator. In Section~\ref{sec:simulation}, we conduct simulation studies to assess the performance of the proposed estimator. In Section~\ref{sec:rda}, we apply the proposed method to a semi-synthetic dataset generated from a publicly available breast cancer dataset. Finally,  concluding remarks are provided in Section~\ref{sec:summary}.

\section{Methodology}
\label{sec:notation} \label{sec:setup}

\subsection{Preliminary on Differential Privacy on Survival Data}

Consider a study with $n$ individuals. Let $T$ and $C$ denote the event time and the censoring time, and we assume that $T$ is independent of $C$ (denoted $C\perp T$). We consider the scenario where the event times are subject to right censoring throughout the paper. Without privacy protection, the observed data are $\mathcal{O}=\{U_{i},\delta_{i}\}_{i=1}^{n}$, where $T_i$ is subject to right censoring,  $U_{i}=\min(T_{i},C_{i})$ and $\delta_{i}=I(T_{i}\le C_{i})$. Define $G_{C}(t)=\Pr(C\ge t)>0$ and $S_{T}(t)=\Pr(T\ge t)$ as the survivor functions of $C$ and $T$, respectively. We are interested in estimating $S_T(t)$ for $t$ in a bounded domain $\mathcal{T}$.

In clinical studies, not only the observed data $\{U_{i},\delta_{i}\}_{i=1}^{n}$ carry high re-identification risk but also the released survivor curve. Differential privacy (DP)
protects against this by randomizing the release mechanism, so that the output
would have been almost as likely had the record been different. The privacy
budget $\epsilon>0$ quantifies the guarantee, with smaller $\epsilon$ meaning
stronger privacy. The original definition for $\epsilon$-differential privacy follows \citep{dwork2006calibrating} is given below.
\begin{definition}[$\epsilon$-differential privacy]
\label{def:dp}
Call two datasets $D,D'$ neighbouring if they differ in one record. A
randomized mechanism $M$ satisfies $\epsilon$-differential privacy if for all
neighbouring $D,D'$ and all measurable sets $A$,
\begin{equation}
  \Pr\{M(D)\in A\}\le {\exp(\epsilon)}\,\Pr\{M(D')\in A\}.
  \label{eq:dp}
\end{equation}
\end{definition}
 
Definition~\ref{def:dp} presumes that the curator can be trusted with the raw data, that is, the raw records themselves remain exposed to the trusted curator and to anyone who compromises it. When no trusted
curator or secure centralized infrastructure is available, this protection is
insufficient. This necessitates the local differential privacy (LDP) framework \citep{kasiviswanathan2011what, duchi2018minimax}, wherein each record is privatized prior to data and no curator ever
holds the true observations $\mathcal{O}$. Formally, the constraint \eqref{eq:dp} is imposed on each
single record, with the two neighbouring datasets replaced by two possible values
of that record.
\begin{definition}[$\epsilon$-local differential privacy]
\label{def:ldp}
Let $W$ denote one individual's record and $Q(\cdot\mid w)$ the conditional probability
of the release $Z$ given $W=w$. The mechanism $Q$ satisfies $\epsilon$-LDP if,
for every measurable set $A$ and every pair of values $w,w'$,
\begin{equation}
  Q(Z\in A\mid W=w)\;\le\;\exp(\epsilon)\,Q(Z\in A\mid W=w').
  \label{eq:ldp}
\end{equation}
When the release has a density $q(\cdot\mid w)$, this is equivalent to
${\exp(-\epsilon)}\le q(z\mid w)/q(z\mid w')\le {\exp(\epsilon)}$ for almost every $z$.
\end{definition}
Local privacy is much more statistically demanding. The central DP in Definition~\ref{def:dp} perturbs an aggregate dataset with sample size $n$ and its effect is simply $O(1/n)$ and thus the convergence rates of any statistical estimators with such an aggregate data are usually unaffected. On the contrary, the local model perturbs each record so it convolves the entire distribution of the data. 

Motivated by Definitions~\ref{def:ldp}, we consider a data releasing setting where each individual releases a privatized version of {$(U_{i},\delta_{i})$}. The
observed time is protected on the log scale, $Y_{i}=\log U_{i}$. %
Privatize $Y_i$ is achieved by adding perturbation as follows: 
\begin{equation}
  Y_{i}^{*}=Y_{i}+e_{i}
  \quad \text{with} \quad
  e_{i}\stackrel{\mathrm{iid}}{\sim}f_e(\cdot)
  \quad \text{and} \quad 
  e_{i}\perp(T_{i},C_{i},X_{i},\delta_{i}),
  \label{eq:perturb-mech}
\end{equation}
where $f_{e}(\cdot)$ is the distribution of perturbation {and `iid' indicates independent and identially distributed}. When $f_e(\cdot)$ is specified as the Laplace distribution with mean zero, we parameterize $f_e(\cdot)$ to be $f_e(e) = \frac{1}{2b} \exp\left(-\frac{\vert{}e\vert{}}{b}\right)$, where {$b \triangleq \frac{\Delta}{\epsilon_{\text{time}}} > 0$} is the scale parameter with $\Delta$ being the sensitivity that indicates the maximum possible difference between any two individual inputs and $\epsilon_{\text{time}}$ is the privacy budget. We assume that $Y_i\in [a_{1},a_{2}]$ and thus $\Delta<\infty$ and $b<\infty$. In practice, $a_{2}=\log\tau$ where $\tau$ is the administrative end of follow-up and $a_{1}=\log t_0$ where $t_0$ can be the resolution of event times at which event times are recorded. By doing this, $\{\{Y^*_i,\delta_i\},i=1,\cdots,n\}$ satisfy $\epsilon_{\mathrm{time}}$-differential privacy because the log-likelihood ratio
$\{|Y^*_i-y'|-|Y^*_i-y|\}/b$ is bounded by $\Delta/b$ \citep{dwork2006calibrating}. Define $\epsilon_{\mathrm{time}}=\Delta/b$ and we call the release of observed event time is
$\epsilon_{\mathrm{time}}$-LDP.

In addition to $Y_i^\ast$ in (\ref{eq:perturb-mech}), the event indicator is privatized by the method of randomized response
\citep{warner1965randomized}. We obtain the privatized event indicator $\delta^*_i$ by randomizing/flipping $\delta_i$ with probability $p_l=\Pr(\delta^{*}_{i}=1\mid\delta_{i}=l)$ for
$l\in\{0,1\}$, with $p_{1}\neq p_{0}$. We can build the connection between the privacy budget and the randomizing probability $p_l$ using the bound of the two ratios $\max\left\{\left|\log\tfrac{p_{1}}{p_{0}}\right|,\
  \left|\log\tfrac{1-p_{1}}{1-p_{0}}\right|\right\}$.
The two ratios $\left|\log\tfrac{p_{1}}{p_{0}}\right|,\
  \left|\log\tfrac{1-p_{1}}{1-p_{0}}\right|$ represent the log-likelihood ratios defined in \eqref{eq:ldp} corresponding to $\delta^*_i=1$ and $0$. A common choice for $p_1, p_0$ is $p_1 = \frac{\exp(\epsilon_{\mathrm{RR}})}{\exp(\epsilon_{\mathrm{RR}})+1}$ and $
p_0 = \frac{1}{\exp(\epsilon_{\mathrm{RR}})+1}$ so that $p_1+p_0=1$ and the two ratios are identical. Therefore, we have that $p_{1}-p_{0}=\tanh(\epsilon_{\mathrm{RR}}/2)$. Since the randomized responses $\delta^*_i$ are generated independently of the perturbations
$e_{i}$ and of $(T_{i},C_{i})$, by the composition
property of differential privacy, releasing the two quantities
$(Y^{*}_{i},\delta^{*}_{i})$ about the same record therefore satisfies
$\epsilon$-LDP with
$\epsilon=\epsilon_{\mathrm{time}}+\epsilon_{\mathrm{RR}}$.

While this perturbation mechanism rigorously guarantees the desired local privacy protection, the injected perturbations inevitably corrupt the underlying data distribution. Consequently, standard survival analysis procedures applied directly to the privatized records $\mathcal{O^*} \triangleq \{Y_i^*, \delta_i^*\}_{i=1}^n$ yields severe bias. For instance, a naive estimator for the survivor function is defined as 
below,{
\begin{equation}
  \widehat{S}_{\mathrm{naive}}(t)
  =\prod_{u\le t}\left\{1-\frac{dN^{*}(u)}{R^{*}(u)}\right\}
  \quad \text{with} \quad
  N^{*}(u)=\sum_{i=1}^{n}\delta^{*}_{i}I({U^{*}_{i}\le u})
  \quad \text{and} \quad
  R^{*}(u)=\sum_{i=1}^{n}I(U^{*}_{i}\ge u).
  \label{eq:naive}
\end{equation}}
This naive estimator is the Kaplan--Meier estimator applied to
$\mathcal{O^*}$ as though the released $U^{*}_{i}=\exp(Y^{*}_{i})$ and
$\delta^{*}_{i}$ were the true observed time and event indicator. It is biased because the true event indicators are diluted by censored records that were wrongly
flagged as events and the independent perturbation $e$ distorts the distribution of the event times. Therefore, estimating $S_{T}(\cdot)$ and valid inference require recovering both event times and event indicators. In the subsequent section, we will describe our proposed method to recover $S_T(\cdot)$ with the LDP released data $\mathcal{O^*}$.


\subsection{Proposed De-biased Estimation of Survivor Functions}
\label{sec:estimator}

We begin our description to identify the target survivor function without privacy protection on the data. With non-privatized data, we define the event sub-density on the log scale as follows:
\begin{equation*}
  g_{1}(v)\triangleq\frac{d}{dv}\Pr\!\left(Y\le v,\ \delta=1\right).
\end{equation*}
Under {the independence assumption} $C\perp T$, the event sub-density and the survivor function of $Y$ are factorized as
\begin{equation*}
  g_{1}(v)=f_{\log T}(v)\,G_{C}(\exp(v))
  \quad \text{and} \quad
   H_{Y}(v)=\Pr(Y\ge v)=S_{T}(\exp(v))\,G_{C}(\exp(v)),
  \label{eq:twofac}
\end{equation*}
where $f_{\log T}(\cdot)$ is the density function of log-transformed event times $\log T$. If the censoring distribution $G_C(\cdot)$ is known, $S_T(t)$ is identified via an inverse-probability-weighting (IPW) representation:
\begin{align}
\label{eq:SurvC-NonPrivate}
S_T(t)=1-F_{T}(t)=1-\int_{-\infty}^{\log t}g_{1}(v)/{G}_{C}(\exp(v))\,dv.\end{align}
When the censoring distribution is unknown, we observe that the ratio $\frac{g_{1}(v)}{H_{Y}(v)}$ is identical to the hazard function of $\log T$. This leads to the Nelson--Aalen type representation for $S_T(t)$:
\begin{align}
\label{eq:SurvNA-NonPrivate}
    S_{T}(t)=\exp\left\{-\int_{-\infty}^{\log t}
  \frac{g_{1}(v)}{H_{Y}(v)}\,dv\right\}.
\end{align}
We then discuss how to identify $S_T(t)$ with the LDP data $\mathcal{O^*}$. To estimate $S_T(t)$ with either \eqref{eq:SurvC-NonPrivate} or \eqref{eq:SurvNA-NonPrivate} under LDP, we must first estimate $g_1(v)$ from privatized data $\mathcal{O^*}$.  

The path from the released data to recover $g_1(\cdot)$ runs through
deconvolution. We assume $f_{\log T}(\cdot)$ and $G_{C}(\cdot)$ are $r\,(\ge2)$ times continuously
differentiable with bounded derivatives on the horizon, so $g_{1}(\cdot)$ inherits the same smoothness. For {any two integrable density functions $f(\cdot)$ and $g(\cdot)$ corresponding to two independent random variables}, the convolution is defined as follows,
\[
  (f*g)(v)=\int_{-\infty}^{\infty}f(y)\,g(v-y)\,dy.
  \label{eq:convdef}
\]
Given $e\perp(Y,\delta)$, the distribution of privatized $Y^*$ can be viewed as the convolution of $Y$ and $e$. We define $g^*_1(v)=\frac{d}{dv}\Pr(Y\leq v, \delta=1)$ and thus $g^*_1(v)=(g_{1}*f_{e})(v)$. Now the target is to de-convolute $g_1(\cdot)$ from $(g_{1}*f_{e})(v)$. Following \citet{stefanski1990deconvolving}, let $\varphi_K(\cdot)$ and ${\varphi_e}(\cdot)$ denote the characteristic functions of a standard smoothing kernel $K(\cdot)$ and the perturbation distribution $f_{e}(\cdot)$, respectively. Since the distribution of perturbation, $f_e(\cdot)$, is generally known when releasing the data, given an estimator $\widehat{\varphi}_K(u)$ for $\varphi_{Y^*}(u)$ with a kernel function $K(\cdot)$ with bandwidth $h>0$, we have that $\widehat{\varphi}_Y(u)=\frac{\widehat{\varphi}_K(u)}{\varphi_e(u/h)}$ and the deconvoluting kernel function for the density function of $Y$ is
\begin{equation}
\mathcal{K}_h(z) = \frac{1}{2\pi h} \int_{-\infty}^{\infty} e^{-\mathrm{i} t z / h} \frac{\varphi_K(t)}{\varphi_e(t / h)} , dt.
\label{eq:deconv-kernel}
\end{equation}
For Laplace$(0, b)$ perturbation, ${\varphi_e}(t) = (1 + b^2 t^2)^{-1}$, and \eqref{eq:deconv-kernel} simplifies analytically to $\mathcal{K}_h(z) = K_h(z) - b^2 K_h''(z)$, where $K_h(z) = h^{-1} K(z/h)$. When paired with a Gaussian base kernel $K(z) = \Phi'(z)$ {with $\Phi(z)$ being the distribution function of the standard normal distribution}, the integral $\overline{\mathcal{K}}_h(z) = \int_{-\infty}^z \mathcal{K}_h(u) du = \Phi(z/h) - (b^2/h^2) K'(z/h)$ is also available in closed form.
When the event indicators are fully available and not flipped, we can recover and estimate $g_1(\cdot)$ using this deconvoluting kernel function by
\begin{align}
\label{eq:tildeg1}
    \widetilde g_{1}(v)=\frac1n\sum_{i=1}^{n}\delta_{i}\,\Kcal_{h}(v-Y^{*}_{i}).
\end{align}
With flipped event indicators, we need to correct the bias. Noting that $\E[\delta^{*}\mid\delta]=p_{0}+(p_{1}-p_{0})\delta$, we propose to reconstruct the uniquely unbiased event indicator by $\dhat_{i}=\frac{\delta^{*}_{i}-p_{0}}{p_{1}-p_{0}}$. By substituting the event indicator $\delta_i$ in \eqref{eq:tildeg1} with the reconstructed $\dhat_{i}$, the proposed estimator for $g_1(v)$ is
\begin{equation}
  \widehat g_{1}(v)=\frac1n\sum_{i=1}^{n}\dhat_{i}\,\Kcal_{h}(v-Y^{*}_{i}).
  \label{eq:ghat}
\end{equation}

When censoring distribution is known, we can plug $\widehat{g_1}(\cdot)$ into \eqref{eq:SurvC-NonPrivate} and obtain an IPW type estimator for the survivor function: 

\begin{equation}\widehat{S}_{\mathrm{IPW}}(t) = 1 - \int_{-\infty}^{\log t} \frac{\widehat{g}_1(v)}{G_C(\exp(v))}dv.\label{eq:SurvIPW-Proposed}\end{equation}In practice, however, the underlying distribution of censoring time is rarely known prior to analysis and thus $G_C(\cdot)$ is not known. With locally private data, $G_C(\cdot)$ is moreover challenging to estimate because it requires a separate deconvolution step to recover the distribution of censoring times. To circumvent the need to model or estimate $G_C(\cdot)$, we adopt the hazard-based formulation in \eqref{eq:SurvNA-NonPrivate} to develop a debiased Nelson–Aalen type estimator.

Following the representation in \eqref{eq:SurvNA-NonPrivate}, we need to estimate $H_Y(v)$ with the LDP data in addition to the event sub-density $g_1(v)$. Using the analogous deconvolution technique for $\widehat{g_1}(v)$, we propose the following smoothed estimator for $H_Y(v)$:
\begin{equation}\widehat{H}_Y(v) = \frac{1}{n} \sum_{i=1}^{n} \int_{v}^{\infty} \mathcal{K}_h(w - Y^{*}_{i})dw.\label{eq:Hhat}\end{equation}
We note that $\widehat H_{Y}(v)$ uses all records irrespective of censoring status and
needs neither the indicator nor its reconstruction. By plugging $\widehat{g}_1(v)$ from \eqref{eq:ghat} and $\widehat{H}_Y(v)$ from \eqref{eq:Hhat} into \eqref{eq:SurvNA-NonPrivate}, we obtain our proposed debiased Nelson–Aalen type estimator for the survivor function 
\begin{equation}
  \widehat S_{\mathrm{NA}}(t)
  =\exp\left\{-\int_{-\infty}^{\log t}
  \frac{\widehat g_{1}(v)}{\widehat H_{Y}(v)}\,dv\right\}.
  \label{eq:SNAhat}
\end{equation}

By the post-processing
property of differential privacy, both $\widehat S_{\mathrm{NA}}(\cdot)$ and $\widehat S_{\mathrm{IPW}}(\cdot)$ inherit the LDP guarantee because they are computed from the privatized data $\mathcal{O}^*$. The estimator $\widehat S_{\mathrm{NA}}$ has advantages that it requires neither specification nor estimation of the distribution of censoring times. Unlike $\widehat{S}_{\mathrm{IPW}}(\cdot)$, this proposed estimator is naturally valued in $(0,1]$ without any additional transformation. Moreover, it is
non-increasing wherever $\widehat g_{1}\ge0$. On the other hand, evaluating the integral in \eqref{eq:SNAhat} for $v$ close to
$\log\tau$ requires cautious, because the estimated survival probability
$\widehat H_{Y}(v)$ in the denominator approaches zero and can fluctuate below zero due to kernel deconvolution noise. To ensure finite-sample stability and avoid division by zero, we restrict evaluation to $t \le \tau'$, where $\tau' < \infty$ satisfies $H_Y(\log \tau') \ge \eta > 0$ for a fixed threshold $\eta$. In this way, the threshold trims the denominator appropriately and makes $\widehat{H}_Y(v)$ distinguishable from zero. 

We estimate the point-wise variance and construct confidence bands for $\widehat{S}_{\mathrm{NA}}(t)$ using a multiplier bootstrap procedure. In the next section, we establish the asymptotic properties of $\widehat{S}_{\mathrm{NA}}(t)$ and the consistency of bootstrap variance estimator of $\widehat{S}_{\mathrm{NA}}(t)$. 
\subsection{Asymptotic Properties}
\label{sec:theory}
We establish asymptotic properties of the proposed estimator 
$\hat{S}_{\mathrm{NA}}(\cdot)$ under the following assumptions.
\begin{assumption}\label{as:ancillary}
The perturbations $e_{i}$ and the randomized/flipped indicators are independent across $i$, and $e_{i}\perp(T_{i},C_{i},\delta_{i})$
 for each $i$. 
\end{assumption}

\begin{assumption}\label{as:noise}
The characteristic function $\varphi_{e}$ for the perturbation never vanishes and
$|\varphi_{e}(u)|\asymp c_{e}|u|^{-\beta}$ as $|u|\to\infty$; for a smoothness
index $\beta\ge1$ and $c_{e}=\lim_{|u|\to\infty}|u|^{\beta}|\varphi_{e}(u)|>0$.
\end{assumption}

\begin{assumption}\label{as:trunc}
The bound $[a_1,a_2]=[\log t_0,\log\tau]$ for privatization is selected independent of the data. The evaluation domain
$\mathcal{T}=[t_{0}',\tau']$ satisfies $t_{0}<t_{0}'<\tau'<\tau$, so that
$\mathcal{L}=\{\log t:t\in\mathcal{T}\}=[\log t_{0}',\log\tau']\subset(a_{1},a_{2})$,
and $H_{Y}(\log\tau')\ge\eta>0$.

\end{assumption}

\begin{assumption}\label{as:smooth}
The $g_{1}$ and $H_{Y}$ are $r$ times continuously
differentiable, $r\ge2$, with derivatives bounded uniformly on an open
neighbourhood of $\mathcal{L}$. In addition the
density of $Y^{*}$ is bounded on $\mathbb{R}$.
\end{assumption}

\begin{assumption}\label{as:kernel}
The kernel function $K$ used to derive the deconvolution kernel $\mathcal{K}_h(\cdot)$ in \eqref{eq:deconv-kernel} satisfies:
(i) $K$ is symmetric with $\mu_{0}(K)=1$ and of order $r\ge2$, that is
$\mu_{j}(K)=0$ for $1\le j<r$ and $\mu_{r}(K)\neq0$, where
$\mu_{j}(K)=\int u^{j}K(u)\,du$, and $\int|u|^{r}|K(u)|\,du<\infty$;
(ii) $K$ is $m$ times continuously differentiable for an integer
$m\ge\lceil\beta\rceil$, with $\int (K^{(j)})^{2}(u)\,du<\infty$ for
$0\le j\le m$;
(iii) $K$ and its first $m$ derivatives are of bounded variation.
\end{assumption}

\begin{assumption}\label{as:bw}
As $n\to\infty$, $h\to0$ and
\[
  \frac{n\,h^{2\beta-1}c_{e}^{2}}{\log n}\longrightarrow\infty .
\]
\end{assumption}

Assumption~\ref{as:noise} is the standard notion of \emph{ordinary smoothness} of order $\beta$ \citep{fan1991optimal}, which requires the
characteristic function of the perturbation to decay at a polynomial rate, so
that dividing by it in \eqref{eq:deconv-kernel} amplifies high frequencies only polynomially. Assumption~\ref{as:trunc} ensures the LDP guarantee with Laplace perturbation and 
restricts estimation of $H_{Y}(\cdot)$ to a pre-specified range on 
$H_{Y}(\cdot)$ is bounded away from zero. Assumption~\ref{as:smooth} requires $g_{1}(\cdot)$ and $H_{Y}(\cdot)$ to be smooth enough. Assumption~\ref{as:kernel} collects the usual conditions on the kernel in
deconvolution estimation \citep{stefanski1990deconvolving,fan1991optimal} to
ensure that $\mathcal{K}_{h}$ is well defined and its integral are uniformly bounded. Assumption~\ref{as:bw} allows the bandwidth to shrink so that the smoothing bias
vanishes and each kernel window retains sufficient
information so that we can asymptotically achieve a consistent estimator. The
Laplace mechanism satisfies Assumption~\ref{as:noise} with $\beta=2$ and $c_{e}=b^{-2}$, since
$\varphi_{e}(u)=(1+b^{2}u^{2})^{-1}\asymp b^{-2}|u|^{-2}$ and a Gaussian kernel satisfies Assumption~\ref{as:kernel} with $r=2$.

The proposed estimator $\hat{S}_{\mathrm{NA}}(\cdot)$ in \eqref{eq:SNAhat} thus has the following asymptotic properties.
\begin{theorem}[Uniform Consistency]
\label{thm:main}

Under the LDP releasing mechanism \eqref{eq:perturb-mech} and
Assumptions~\ref{as:ancillary}--\ref{as:bw},
\begin{equation}
  \sup_{t\in\mathcal{T}}\bigl|\widehat S_{\mathrm{NA}}(t)-S_{T}(t)\bigr|
  =O_{p}(\rho_{n})
  \quad \text{with} \quad
  \rho_{n}=h^{r}+\Bigl(\frac{\log n}{c_{e}^{2}\,n\,h^{2\beta-1}}\Bigr)^{1/2}.
  \label{eq:unif}
\end{equation}
For the Laplace mechanism, $\beta=2$, $c_{e}=b^{-2}$,
$\rho_{n}=h^{r}+\{b^{4}\log n/(nh^{3})\}^{1/2}$. 
Balancing the two terms of
$\rho_{n}$ gives
\begin{equation}
  h\asymp\Bigl(\frac{b^4\log n}{n}\Bigr)^{\frac{1}{2r+3}}
  \quad \text{and} \quad
  \sup_{t\in\mathcal{T}}\bigl|\widehat S_{\mathrm{NA}}(t)-S_{T}(t)\bigr|
  =O_{p}\Bigl\{\Bigl(\frac{b^4\log n}{n}\Bigr)^{\frac{r}{2r+3}}\Bigr\}.
  \label{eq:optrate}
\end{equation}
In particular, under a fixed privacy budget $\epsilon_{\mathrm{time}}>0$, so
that $b=\Delta/\epsilon_{\mathrm{time}}$ is constant in $n$, $\rho_{n}\to0$ and
the estimator is uniformly consistent,
\begin{equation}
  \sup_{t\in\mathcal{T}}\bigl|\widehat S_{\mathrm{NA}}(t)-S_{T}(t)\bigr|
  \overset{p}{\longrightarrow}0.
  \label{eq:S_NA_op1}
\end{equation}

\end{theorem}

\begin{theorem}[Asymptotic Normality and Consistency of Variance Estimation]
\label{thm:boot}
\ \\

Let
$\varsigma_{n}^{2}(t)=n\Var\{\widehat S_{\mathrm{NA}}(t)\}$ for a fixed $t$ in $\mathcal{T}$. Under
Assumptions~\ref{as:ancillary}--\ref{as:bw},
\begin{equation}
  \sqrt{n}\,\varsigma_{n}^{-1}(t)
  \left\{\widehat S_{\mathrm{NA}}(t)-\E\widehat S_{\mathrm{NA}}(t)\right\}
  \stackrel{d}{\longrightarrow} N(0,1) \quad \text{as} \quad n \rightarrow \infty. 
  \label{eq:clt}
\end{equation}
Let $\omega_{1},\dots,\omega_{n}$ be independent and identically distributed,
independent of the data, with
{$\E(\omega_j)=1$ and $\Var(\omega_j)=1$ for all $j=1,\cdots,n$.} Let
$\widehat S^{\omega}_{\mathrm{NA}}(t)$ be recomputed from $\widehat g^{\omega}_{1}(v)
  =\frac{1}{n}\sum_{i=1}^{n}\omega_{i}\,\dhat_{i}\,\Kcal_{h}(v-Y^{*}_{i})$ and $
  \widehat H^{\omega}_{Y}(v)
  =\frac{1}{n}\sum_{i=1}^{n}\omega_{i}
   \int_{v}^{\infty}\Kcal_{h}(w-Y^{*}_{i})\,dw$.
Then,
conditionally on $\mathcal{O}^*$, the conditional variance
$\widehat\varsigma_{n}^{2}(t)
=n\Var\{\widehat S^{\omega}_{\mathrm{NA}}(t)\mid\ \mathcal{O}^*\}$ satisfies 
\begin{equation}
  \widehat\varsigma_{n}^{2}(t)\big/\varsigma_{n}^{2}(t)
  \overset{p}{\longrightarrow}1 .
  \label{eq:bootvar}
\end{equation}
Therefore, $\widehat\varsigma_{n}^{2}(t)$ consistently estimate $\varsigma_n^2(t)$.
\end{theorem}
We outline proofs for Theorems~\ref{thm:main} and~\ref{thm:boot} in Supplementary Mateiral and we also have several remarks on the theorems.
\begin{remark}\label{rm:rate}
    When the perturbation distribution is Laplace, the smoothness index $\beta=2$. With the order of Kernel function, $r=2$, the rate of uniform convergence is $O_p\left\{ \left( \frac{\log n}{n} \right)^{2/7} \right\}$.
\end{remark}
\begin{remark}\label{rm:samplesize}
Since $b=\Delta/\epsilon_{\mathrm{time}}$, the uniform convergence rate in
\eqref{eq:optrate} depends on both the sample size $n$ and the privacy budget $\epsilon_{\text{time}}$ through the combination $n\epsilon_{\mathrm{time}}^{4}/(\Delta^{4}\log n)$.
Hence, if the privacy budget is tightened, the sample size must increase
accordingly to maintain the same bound on the uniform convergence rate. In
particular, halving $\epsilon_{\mathrm{time}}$ requires about sixteen times as
many records.    
\end{remark}
\begin{remark}\label{rm:bw}
    Theorem~\ref{thm:boot} centers the asymptotic distribution of $\widehat{S}_{\mathrm{NA}}(t)$ at $\mathbb{E}\widehat{S}_{\mathrm{NA}}(t)$ rather than at the true target $S_{T}(t)$. This is because $\widehat{S}_{\mathrm{NA}}(t)$ is a smoothed estimator and the bias due to smoothing is not corrected and is beyond scope of this approach. Although the uniform bandwidth in Theorem~\ref{thm:main} optimizes estimation error by balancing both the smoothing biases and the stochastic fluctuations, it does not minimize the point-wise variance against squared bias. In practice, we recommend an optimal bandwidth $h_{\mathrm{opt}} \asymp \left( \frac{b^4}{n} \right)^{\frac{1}{2r+3}}$ that minimizes the mean squared error (MSE). It can be verified that $h_{\mathrm{opt}}$ also guarantees uniform consistency over $\mathcal{T}$.
\end{remark}

\section{Simulation Studies} \label{sec:simulation}
We evaluate the proposed estimator, the Nelson--Aalen-type
estimator $\widehat{S}_{\mathrm{NA}}(\cdot)$ in \eqref{eq:SNAhat}, which does not require the specification of censoring distribution and compared it with several counterparts including $\widehat{S}_{\mathrm{IPW}}(\cdot)$ in \eqref{eq:SurvIPW-Proposed} with a known censoring distribution, and a naive estimator, denoted $\widehat{S}_{\mathrm{naive}}(\cdot)$, is given in \eqref{eq:naive}
that does not correct the biases from privatization.

We conduct two simulation studies. Simulation I examined the finite-sample performance of the proposed estimator in consistency and assessed the
performance of the variance estimators. Simulation II assessed the validity of optimal bandwidth $h$ given in Remark~\ref{rm:bw} with a fixed privacy budget.
\subsection{Simulation I}

We generate data as follows. For each subject $i$, we generated true event and censoring times independently as
$\log T_{i}\sim N(1,0.5^{2})$ and
$\log C_{i}\sim N(\mu_C,0.7^{2})$. We chose different values of $\mu_C=1.218$ to generate a moderate censoring rate with $40\%$ in the data and $\mu_C=0.782$ with heavy censoring rate around $60\%$. Following \eqref{eq:perturb-mech}, we generated privatized observed time subject to censoring with $e_{i}\sim\mathrm{Laplace}(0,b)$. We then flipped the event indicator by $\Pr(\delta_{i}^{*}=1\mid\delta_{i})
=p_1\delta_{i}+p_0(1-\delta_{i})$.
We considered three privacy budgets for protecting the event times according to the deployed local differential privacy systems by Apple Inc. Apple reports per-donation budgets of $\epsilon_{\mathrm{time}}=2$ for Health
data types, $4$ for emoji suggestions and $8$ for QuickType
suggestions\footnote{\url{https://www.apple.com/privacy/docs/Differential_Privacy_Overview.pdf}}. We refer to these
three settings as strong, moderate and mild privacy. We chose the release domain for the data to be $[a_{1},a_{2}]=[-0.7,2.2]$ so that $\Delta=2.9$ and $b=\Delta/\epsilon_{\mathrm{time}}
=1.450,0.725,0.363$ with different privacy budgets. We chose $\epsilon_{\mathrm{RR}}=2$, so that $p_{1}=\exp(2)/\{1+\exp(2)\}=0.881$ and
$p_{0}=1/\{1+\exp(2)\}=0.119$ to flip the event indicators. This lies within the range implied by the parameters of Google's
RAPPOR system \citep{erlingsson2014rappor}. We considered
sample sizes $n=1000, 4000$ and replicated each simulation scenario for 300 times. 

We employed the Gaussian kernel for $K_h(\cdot)$ with the bandwidth
$h=(b^{4}/n)^{1/7}$ (with $r=2$) as suggested by Remark~\ref{rm:bw} to implement proposed estimators and the naive estimator. For each replicate, the standard error of
$\widehat{S}_{\mathrm{NA}}(\cdot)$ was estimated by the multiplier bootstrap of
Theorem~\ref{thm:boot} with $B=150$ draws of $\omega_{i}\sim\mathrm{Exp}(1)$. By choosing the trimming threshold $\eta=0.05$ which ensures at least 5\% of subjects are at risk, we evaluated all estimators on $t\in[1,\exp(1.5)]$. We also assessed the accuracy of estimators by computing the mean absolute errors (MAE) on a grid of $K$ time points, which is $\frac{1}{K}\sum_{k=1}^{K} \left\vert{} \widehat{S}_{est}(t_k) - S(t_k) \right\vert{}$
for any estimator $\widehat{S}_{est} (t)$ with the true survivor function $S_T(t)$ of $T$. 

Tables~\ref{tab:sim1point1000} and \ref{tab:sim1point4000} summarize estimates for $\widehat{S}_{\mathrm{NA}}(t),\widehat{S}_{\mathrm{IPW}}(t)$ and $\widehat{S}_{\mathrm{naive}}(t)$ evaluated $\log t=0.3,0.7,1.1$ over the $300$ replicates. The sample mean of point-wise estimated survivor function, the bias, the sample standard error (SSE) and
the mean of the estimated standard errors obtained from the multiplier
bootstrap (ESE) are reported. In addition to evaluating the true survivor function for $T$, we also computed an oracle
kernel estimator 
\begin{eqnarray} \label{S_oracle}
\widehat{S}_{\mathrm{oracle}}(t)=\exp\Big\{-\int_{-\infty}^{\log t}
  \frac{n^{-1}\sum_{i}\delta_{i}K_{h}(v-Y_{i})}
       {n^{-1}\sum_{i}\int_{v}^\infty K_h(w-Y_i)dw}\Big\}
\end{eqnarray}
that is obtained with non-private data $\{\{Y_{i},\delta_{i}\},i=1,\cdots,n\}$. Comparing with the proposed estimator $\widehat{S}_{\mathrm{NA}}(\cdot)$ in \eqref{eq:SNAhat}, this oracle kernel estimator substitute $\Kcal_h(\cdot)$ with the ordinary kernel
$K_{h}$.

We see that the sample means of our proposed estimators $\widehat{S}_{\mathrm{NA}}(t)$ with unknown censoring distribution and $\widehat{S}_{\mathrm{IPW}}(t)$ with known censoring distribution are fairly close to the true survivor functions and the oracle kernel estimator. The closeness of the proposed estimators to the oracle kernel estimator reveals that the biases incurred in the proposed estimators are induced by the intrinsic smoothing bias and our proposed estimator has de-biased the perturbation in the resulting estimators. This agreement between proposed estimator and the oracle kernel estimator holds regardless of how strong the privacy protection is, which confirms that our proposed estimator can perform well to de-bias even when privacy budget is tightened. It also verifies that the bandwidth we chose makes the stochastic error term vanish and the consistency of our proposed estimator holds, {verifying the validity of (\ref{eq:S_NA_op1}) in Theorem~\ref{thm:main}.}

As expected, the bias and the SSE both decrease as the
privacy budget loosens and as $n$ increases, which is consistent with the
uniform consistency established in Theorem~\ref{thm:main}. Comparing $\widehat{S}_{\mathrm{NA}}(t)$ with $\widehat{S}_{\mathrm{IPW}}(t)$, their bias and SSE are almost identical to each other, indicating that loss of efficiency due to not
knowing $G_{C}(\cdot)$ can thus be recovered by $\widehat{S}_{\mathrm{NA}}(t)$. The agreement is closest under mild privacy and moderate censoring and remains close under strong privacy and heavy censoring, where $\widehat{S}_{\mathrm{NA}}(t)$ carries a slightly larger bias because $\widehat{H}_{Y}(v)$ becomes small in the right tail. However, even in that most demanding scenario, the penalty for not knowing the censoring distribution is modest.

Furthermore, the mean bootstrap standard error agrees
closely with the sample standard error for both 
$\widehat{S}_{\mathrm{NA}}(t)$ and $\widehat{S}_{\mathrm{IPW}}(t)$ across all scenarios of privacy budgets and censoring. This verifies that the multiplier bootstrap of Theorem~\ref{thm:boot} provides a consistent
variance estimator. On the contrary, the naive estimator exhibit substantial biases with 3 to 4 times of the biases for $\widehat{S}_{\mathrm{NA}}(t)$. On the other hand, it is noted that SSE and ESE from $\widehat{S}_{\mathrm{naive}}(t)$ are smaller than those from $\widehat{S}_{\mathrm{NA}}(t)$ and $\widehat{S}_{\text{IPW}}$. This is as expected because higher SSE and ESE are price of bias correction. However, combining biases and SSEs show that the resulting MSE of the naive estimator are larger than those of $\widehat{S}_{\mathrm{NA}}(t)$ and $\widehat{S}_{\mathrm{IPW}}(t)$ at
$\log t=0.3$ and $1.1$, exhibiting the evidence of dealing with differential privacy. The
exception is $\log t=0.7$, where the bias of the naive estimator changes sign
and is therefore close to zero. This reflects the choice of evaluation point
rather than any property of the de-biased estimator.

\begin{table}[htbp]
\caption{Simulation I: Sample mean, bias, sample standard error (SSE) and mean estimated standard error (ESE) from the multiplier bootstrap for proposed estimates over 300 replicates with $n=1000$}
\label{tab:sim1point1000}
\scalebox{0.75}{
\hspace{-1cm}
\begin{tabular}{llll l llll llll llll}\toprule
& & & & $\widehat{S}_{\mathrm{oracle}}$ & \multicolumn{4}{c}{$\widehat{S}_{\mathrm{NA}}$} & \multicolumn{4}{c}{$\widehat{S}_{\mathrm{IPW}}$} & \multicolumn{4}{c}{$\widehat{S}_{\mathrm{naive}}$}\\
\cmidrule(lr){5-5}\cmidrule(lr){6-9}\cmidrule(lr){10-13}\cmidrule(lr){14-17}
privacy & censoring & $\log t$ & $S_T(t)$ & mean & mean & bias & SSE & ESE & mean & bias & SSE & ESE & mean & bias & SSE & ESE\\ \midrule
strong & moderate & 0.3 & 0.919 & 0.863 & 0.859 & -0.060 & 0.046 & 0.044 & 0.862 & -0.057 & 0.045 & 0.044 & 0.765 & -0.154 & 0.014 & 0.014\\
 &  & 0.7 & 0.726 & 0.705 & 0.705 & -0.021 & 0.054 & 0.050 & 0.712 & -0.014 & 0.053 & 0.051 & 0.689 & -0.037 & 0.017 & 0.016\\
 &  & 1.1 & 0.421 & 0.494 & 0.500 & 0.079 & 0.061 & 0.059 & 0.501 & 0.081 & 0.062 & 0.063 & 0.603 & 0.182 & 0.017 & 0.017\\
\addlinespace
 & heavy & 0.3 & 0.919 & 0.878 & 0.876 & -0.044 & 0.044 & 0.043 & 0.880 & -0.039 & 0.044 & 0.042 & 0.799 & -0.120 & 0.014 & 0.014\\
 &  & 0.7 & 0.726 & 0.734 & 0.739 & 0.013 & 0.056 & 0.055 & 0.745 & 0.019 & 0.054 & 0.055 & 0.733 & 0.007 & 0.016 & 0.016\\
 &  & 1.1 & 0.421 & 0.535 & 0.544 & 0.124 & 0.076 & 0.075 & 0.528 & 0.107 & 0.082 & 0.081 & 0.657 & 0.236 & 0.018 & 0.018\\
\midrule
moderate & moderate & 0.3 & 0.919 & 0.889 & 0.889 & -0.031 & 0.025 & 0.027 & 0.890 & -0.029 & 0.025 & 0.028 & 0.817 & -0.102 & 0.012 & 0.013\\
 &  & 0.7 & 0.726 & 0.713 & 0.714 & -0.011 & 0.034 & 0.034 & 0.719 & -0.007 & 0.034 & 0.036 & 0.701 & -0.025 & 0.016 & 0.016\\
 &  & 1.1 & 0.421 & 0.464 & 0.465 & 0.044 & 0.042 & 0.041 & 0.464 & 0.043 & 0.045 & 0.046 & 0.556 & 0.135 & 0.018 & 0.018\\
\addlinespace
 & heavy & 0.3 & 0.919 & 0.897 & 0.897 & -0.023 & 0.027 & 0.028 & 0.899 & -0.021 & 0.026 & 0.028 & 0.831 & -0.089 & 0.012 & 0.013\\
 &  & 0.7 & 0.726 & 0.730 & 0.731 & 0.005 & 0.038 & 0.038 & 0.733 & 0.007 & 0.038 & 0.040 & 0.723 & -0.003 & 0.017 & 0.016\\
 &  & 1.1 & 0.421 & 0.488 & 0.492 & 0.071 & 0.055 & 0.053 & 0.479 & 0.059 & 0.061 & 0.061 & 0.592 & 0.171 & 0.020 & 0.019\\
\midrule
mild & moderate & 0.3 & 0.919 & 0.904 & 0.905 & -0.014 & 0.018 & 0.017 & 0.906 & -0.013 & 0.018 & 0.018 & 0.868 & -0.051 & 0.011 & 0.011\\
 &  & 0.7 & 0.726 & 0.720 & 0.720 & -0.006 & 0.024 & 0.025 & 0.721 & -0.004 & 0.025 & 0.027 & 0.714 & -0.012 & 0.015 & 0.015\\
 &  & 1.1 & 0.421 & 0.444 & 0.445 & 0.025 & 0.031 & 0.030 & 0.443 & 0.023 & 0.034 & 0.037 & 0.504 & 0.083 & 0.020 & 0.018\\
\addlinespace
 & heavy & 0.3 & 0.919 & 0.908 & 0.909 & -0.010 & 0.018 & 0.019 & 0.910 & -0.009 & 0.018 & 0.019 & 0.863 & -0.056 & 0.011 & 0.012\\
 &  & 0.7 & 0.726 & 0.729 & 0.728 & 0.003 & 0.029 & 0.028 & 0.729 & 0.003 & 0.030 & 0.031 & 0.712 & -0.013 & 0.017 & 0.017\\
 &  & 1.1 & 0.421 & 0.458 & 0.458 & 0.037 & 0.042 & 0.039 & 0.448 & 0.028 & 0.047 & 0.048 & 0.513 & 0.092 & 0.022 & 0.021\\
 \bottomrule\end{tabular}}\end{table}

\begin{table}[htbp]
\caption{Simulation I: Sample mean, bias, sample standard error (SSE) and mean estimated standard error (ESE) from the multiplier bootstrap for proposed estimates over 300 replicates with $n=4000$}
\label{tab:sim1point4000}
\scalebox{0.75}{
\hspace{-1cm}
\begin{tabular}{llll l llll llll llll}\toprule
& & & & $\widehat{S}_{\mathrm{oracle}}(\cdot)$ & \multicolumn{4}{c}{$\widehat{S}_{\mathrm{NA}}(\cdot)$} & \multicolumn{4}{c}{$\widehat{S}_{\mathrm{IPW}}(\cdot)$} & \multicolumn{4}{c}{$\widehat{S}_{\mathrm{naive}}(\cdot)$}\\
\cmidrule(lr){5-5}\cmidrule(lr){6-9}\cmidrule(lr){10-13}\cmidrule(lr){14-17}
privacy & censoring & $\log t$ & $S_T(t)$ & mean & mean & bias & SSE & ESE & mean & bias & SSE & ESE & mean & bias & SSE & ESE\\ \midrule
strong & moderate & 0.3 & 0.919 & 0.878 & 0.881 & -0.038 & 0.029 & 0.029 & 0.883 & -0.036 & 0.030 & 0.029 & 0.766 & -0.154 & 0.007 & 0.007\\
 &  & 0.7 & 0.726 & 0.709 & 0.712 & -0.013 & 0.033 & 0.033 & 0.718 & -0.008 & 0.034 & 0.033 & 0.689 & -0.037 & 0.008 & 0.008\\
 &  & 1.1 & 0.421 & 0.478 & 0.477 & 0.056 & 0.038 & 0.039 & 0.477 & 0.056 & 0.040 & 0.042 & 0.601 & 0.181 & 0.009 & 0.009\\
\addlinespace
 & heavy & 0.3 & 0.919 & 0.888 & 0.890 & -0.029 & 0.029 & 0.028 & 0.893 & -0.027 & 0.029 & 0.028 & 0.799 & -0.120 & 0.007 & 0.007\\
 &  & 0.7 & 0.726 & 0.732 & 0.732 & 0.007 & 0.034 & 0.036 & 0.735 & 0.010 & 0.034 & 0.037 & 0.732 & 0.006 & 0.008 & 0.008\\
 &  & 1.1 & 0.421 & 0.510 & 0.508 & 0.088 & 0.049 & 0.051 & 0.493 & 0.072 & 0.052 & 0.055 & 0.655 & 0.235 & 0.009 & 0.009\\
\midrule
moderate & moderate & 0.3 & 0.919 & 0.898 & 0.896 & -0.023 & 0.017 & 0.017 & 0.897 & -0.022 & 0.017 & 0.017 & 0.816 & -0.103 & 0.006 & 0.006\\
 &  & 0.7 & 0.726 & 0.716 & 0.716 & -0.010 & 0.022 & 0.021 & 0.718 & -0.007 & 0.022 & 0.022 & 0.700 & -0.026 & 0.008 & 0.008\\
 &  & 1.1 & 0.421 & 0.451 & 0.453 & 0.032 & 0.026 & 0.026 & 0.451 & 0.031 & 0.028 & 0.029 & 0.555 & 0.135 & 0.009 & 0.009\\
\addlinespace
 & heavy & 0.3 & 0.919 & 0.902 & 0.900 & -0.019 & 0.018 & 0.018 & 0.902 & -0.018 & 0.017 & 0.018 & 0.830 & -0.090 & 0.006 & 0.006\\
 &  & 0.7 & 0.726 & 0.728 & 0.730 & 0.005 & 0.024 & 0.024 & 0.732 & 0.006 & 0.025 & 0.025 & 0.722 & -0.004 & 0.008 & 0.008\\
 &  & 1.1 & 0.421 & 0.469 & 0.472 & 0.051 & 0.034 & 0.035 & 0.460 & 0.039 & 0.038 & 0.038 & 0.590 & 0.169 & 0.010 & 0.010\\
\midrule
mild & moderate & 0.3 & 0.919 & 0.909 & 0.910 & -0.009 & 0.010 & 0.010 & 0.910 & -0.009 & 0.010 & 0.011 & 0.868 & -0.052 & 0.005 & 0.006\\
 &  & 0.7 & 0.726 & 0.721 & 0.722 & -0.004 & 0.014 & 0.015 & 0.723 & -0.003 & 0.015 & 0.015 & 0.713 & -0.012 & 0.008 & 0.008\\
 &  & 1.1 & 0.421 & 0.437 & 0.434 & 0.014 & 0.017 & 0.018 & 0.434 & 0.013 & 0.019 & 0.021 & 0.503 & 0.082 & 0.009 & 0.009\\
\addlinespace
 & heavy & 0.3 & 0.919 & 0.911 & 0.912 & -0.007 & 0.011 & 0.011 & 0.912 & -0.007 & 0.011 & 0.011 & 0.862 & -0.057 & 0.006 & 0.006\\
 &  & 0.7 & 0.726 & 0.728 & 0.728 & 0.002 & 0.016 & 0.017 & 0.729 & 0.003 & 0.017 & 0.018 & 0.712 & -0.013 & 0.008 & 0.008\\
 &  & 1.1 & 0.421 & 0.446 & 0.444 & 0.023 & 0.022 & 0.024 & 0.438 & 0.018 & 0.025 & 0.028 & 0.512 & 0.091 & 0.010 & 0.011\\
\bottomrule\end{tabular}}\end{table}

Table~\ref{tab:sim1mae} reports the sample mean and standard deviation of the MAEs, and Figures~\ref{fig:mae1000} and \ref{fig:mae4000} display the corresponding boxplots for $n=1000$ and $4000$. The proposed
estimator $\widehat{S}_{\mathrm{NA}}$ is close to $\widehat{S}_{\mathrm{IPW}}(t)$ throughout all simulation scenarios and both of them show smaller MAEs than the naive estimator. The difference becomes more evident when privacy gets stronger. It's also notable that the mean MAE of the naive estimator remains unchanged with different sample sizes, confirming that the naive estimator is inconsistent and the necessity to propose a de-bias procedure.

Figures~\ref{fig:curves1000} and \ref{fig:curves4000} plot the estimated survival curves against $t$,
together with the point-wise $2.5\%$ and $97.5\%$ empirical percentiles of
$\widehat{S}_{\mathrm{NA}}(\cdot)$ over the $300$ replicates. The proposed estimator $\widehat{S}_{\mathrm{NA}}(\cdot)$
and $\widehat{S}_{\mathrm{IPW}}(\cdot)$ are nearly indistinguishable and track the true curve
closely. While both of them deviates from the true curve in the right tail, the approximated 95\% confidence bands still cover the true curve. The naive
curve is visibly flattened and deviated from the truth. Analogous to what we have observed in the MAEs, the naive curve does not improve as $n$ increases and it lies outside the empirical band of the proposed estimator over most of the range particularly for a strong privacy scenario.

\begin{table}[htbp]\centering\small
\caption{Simulation I: sample mean and standard deviation of the mean absolute error (MAE)}
\label{tab:sim1mae}
\begin{tabular}{rll cccc}\toprule
& & & $\widehat{S}_{\mathrm{oracle}}(\cdot)$ & $\widehat{S}_{\mathrm{NA}}(\cdot)$ & $\widehat{S}_{\mathrm{IPW}}(\cdot)$ & $\widehat{S}_{\mathrm{naive}}(\cdot)$\\
\cmidrule(lr){4-4}\cmidrule(lr){5-5}\cmidrule(lr){6-6}\cmidrule(lr){7-7}
$n$ & privacy & censoring & mean (SD) & mean (SD) & mean (SD) & mean (SD)\\ \midrule
1,000 & strong & moderate & 0.0602 (0.0029) & 0.0737 (0.0294) & 0.0692 (0.0296) & 0.1600 (0.0053)\\
 &  & heavy & 0.0719 (0.0051) & 0.0874 (0.0351) & 0.0758 (0.0351) & 0.1708 (0.0062)\\
 & moderate & moderate & 0.0340 (0.0043) & 0.0416 (0.0155) & 0.0394 (0.0163) & 0.1124 (0.0059)\\
 &  & heavy & 0.0415 (0.0067) & 0.0513 (0.0195) & 0.0460 (0.0196) & 0.1234 (0.0071)\\
 & mild & moderate & 0.0190 (0.0052) & 0.0252 (0.0092) & 0.0250 (0.0091) & 0.0633 (0.0067)\\
 &  & heavy & 0.0235 (0.0074) & 0.0311 (0.0123) & 0.0309 (0.0126) & 0.0701 (0.0078)\\
\midrule
4,000 & strong & moderate & 0.0455 (0.0018) & 0.0493 (0.0178) & 0.0459 (0.0175) & 0.1589 (0.0026)\\
 &  & heavy & 0.0551 (0.0030) & 0.0602 (0.0218) & 0.0508 (0.0204) & 0.1697 (0.0029)\\
 & moderate & moderate & 0.0242 (0.0022) & 0.0292 (0.0097) & 0.0268 (0.0097) & 0.1120 (0.0030)\\
 &  & heavy & 0.0295 (0.0033) & 0.0361 (0.0121) & 0.0307 (0.0112) & 0.1229 (0.0034)\\
 & mild & moderate & 0.0125 (0.0027) & 0.0153 (0.0050) & 0.0149 (0.0050) & 0.0625 (0.0030)\\
 &  & heavy & 0.0155 (0.0040) & 0.0197 (0.0069) & 0.0186 (0.0065) & 0.0698 (0.0037)\\
\bottomrule\end{tabular}\end{table}
\begin{figure}[htbp]
  \centering
  \begin{subfigure}[t]{0.48\textwidth}
    \centering
    \includegraphics[width=\linewidth]{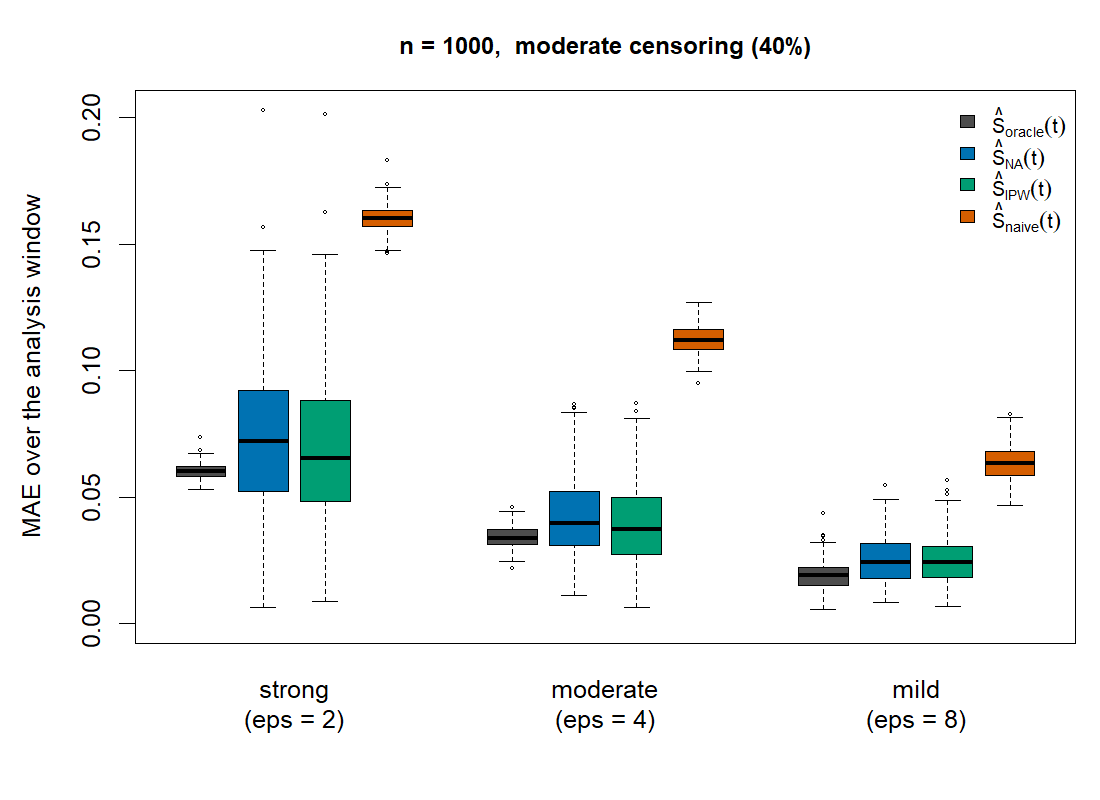}
    \caption{Moderate censoring ($40\%$)}
    \label{fig:mae1000mod}
  \end{subfigure}
  \hfill
  \begin{subfigure}[t]{0.48\textwidth}
    \centering
    \includegraphics[width=\linewidth]{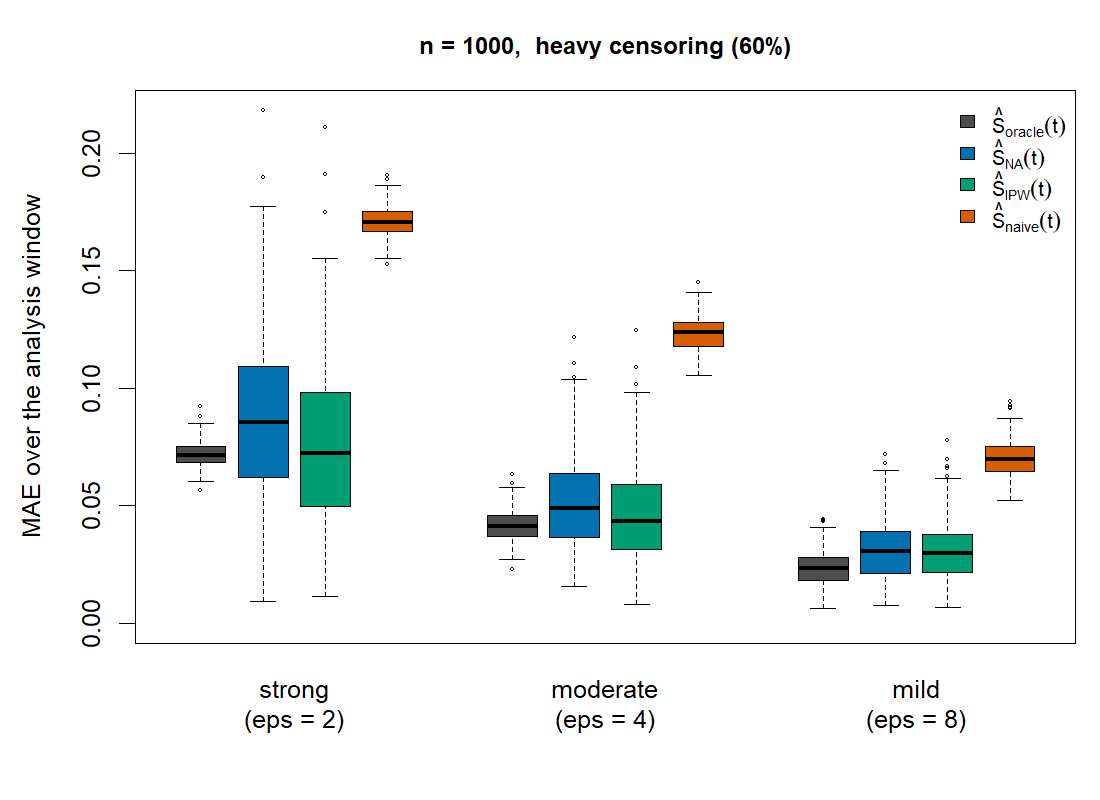}
    \caption{Heavy censoring ($60\%$)}
    \label{fig:mae1000heavy}
  \end{subfigure}
  \caption{Simulation I: Boxplots of the mean absolute error (MAE) for survivor functions over $300$ replicates with $n=1000$}
  \label{fig:mae1000}
\end{figure}

\begin{figure}[htbp]
  \centering
  \begin{subfigure}[t]{0.48\textwidth}
    \centering
    \includegraphics[width=\linewidth]{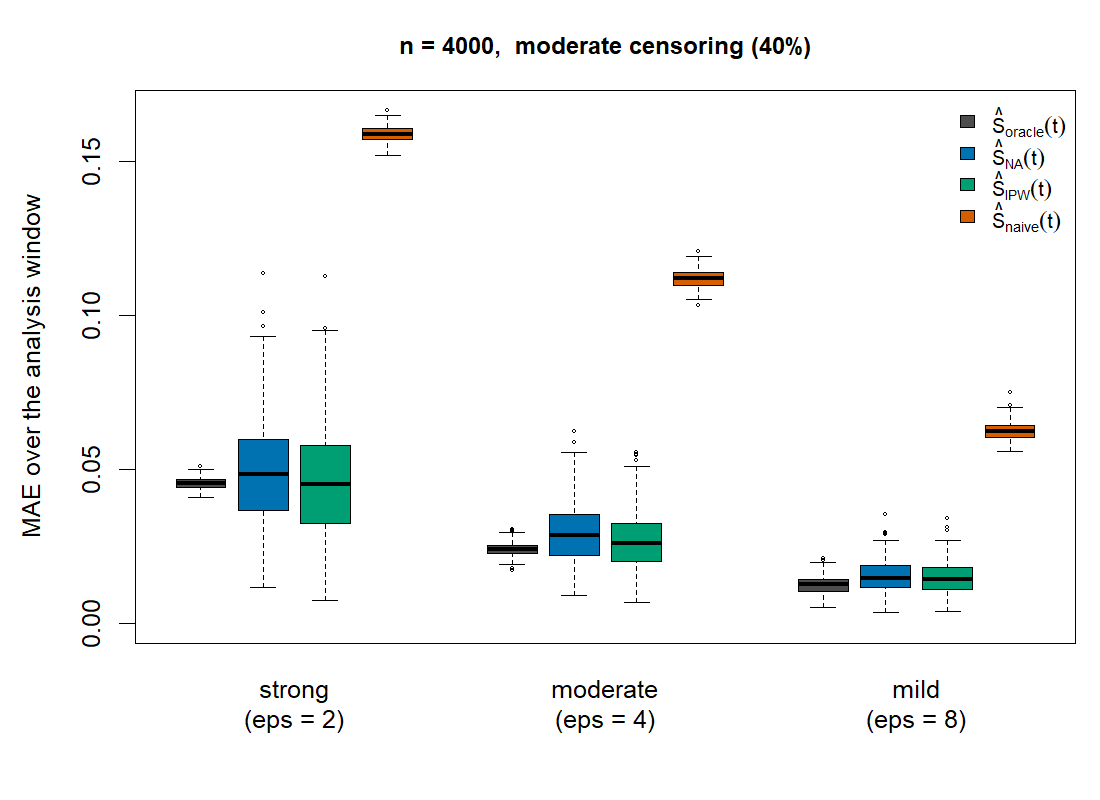}
    \caption{Moderate censoring ($40\%$)}
    \label{fig:mae4000mod}
  \end{subfigure}
  \hfill
  \begin{subfigure}[t]{0.48\textwidth}
    \centering
    \includegraphics[width=\linewidth]{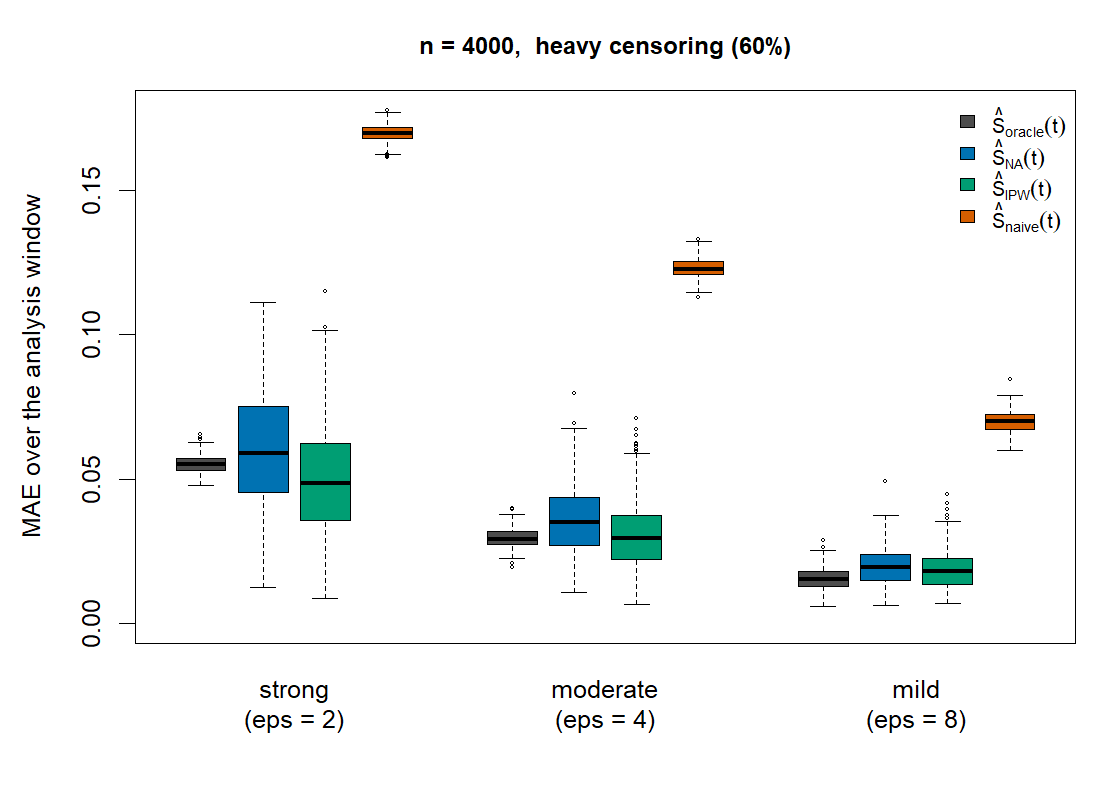}
    \caption{Heavy censoring ($60\%$)}
    \label{fig:mae4000heavy}
  \end{subfigure}
  \caption{Simulation I: Boxplots of the mean absolute error (MAE) for survivor functions over $300$ replicates with $n=4000$}
  \label{fig:mae4000}
\end{figure}

\begin{figure}[htbp]
  \centering
  \begin{subfigure}[t]{0.47\textwidth}
    \centering
    \includegraphics[width=\linewidth]{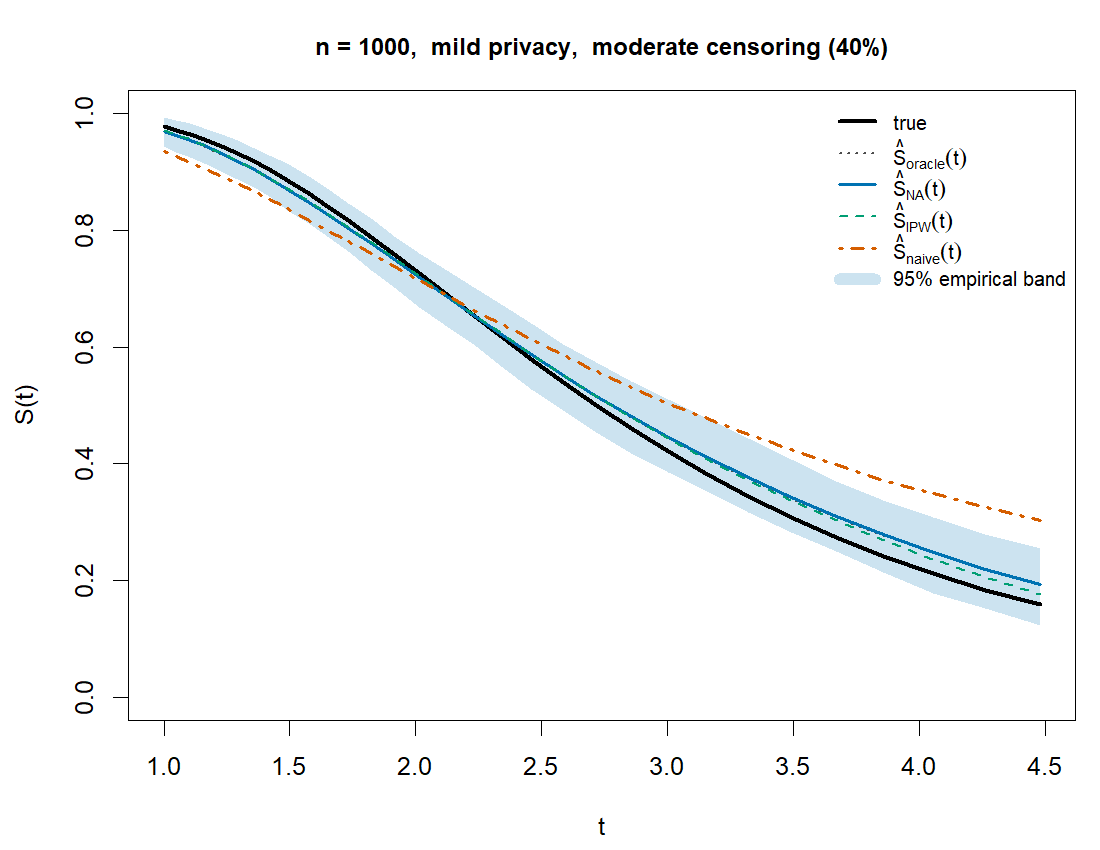}
    \caption{Mild privacy ($\epsilon_{\mathrm{time}}=8$), moderate censoring}
  \end{subfigure}
  \hfill
  \begin{subfigure}[t]{0.47\textwidth}
    \centering
    \includegraphics[width=\linewidth]{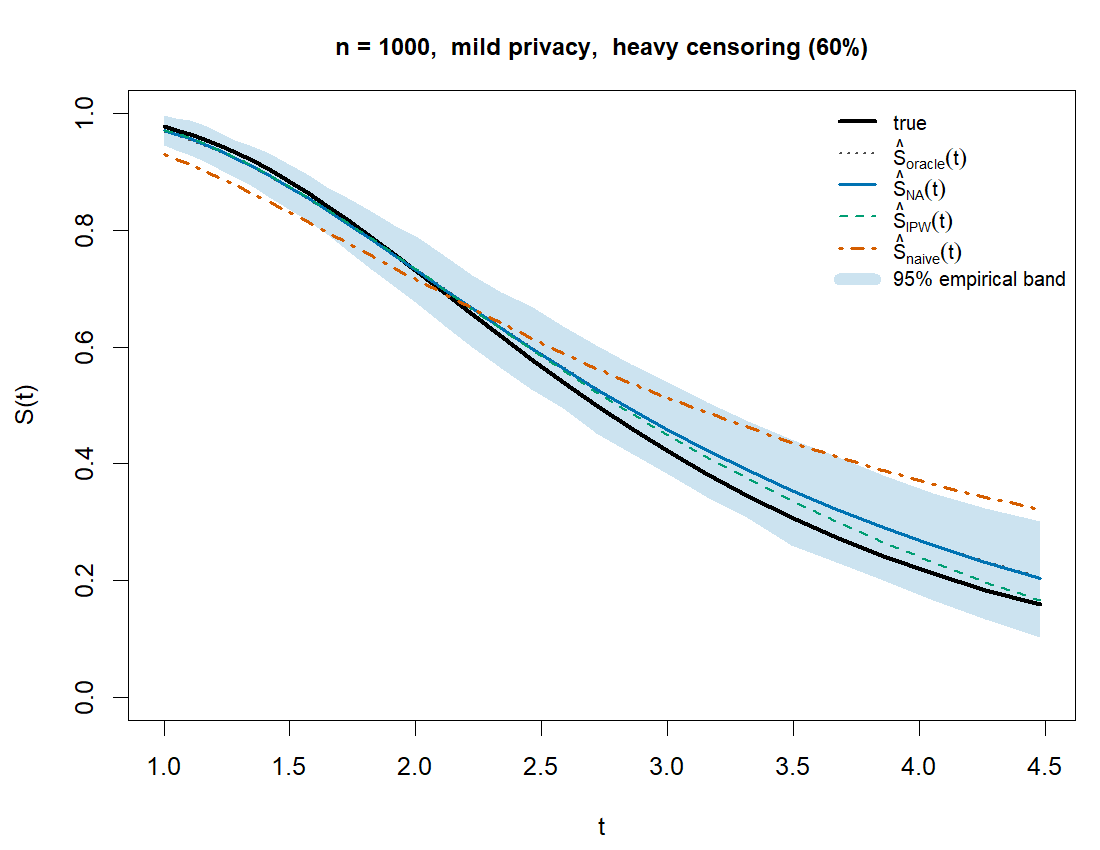}
    \caption{Mild privacy ($\epsilon_{\mathrm{time}}=8$), heavy censoring}
  \end{subfigure}

  \vspace{0.6em}
  \begin{subfigure}[t]{0.47\textwidth}
    \centering
    \includegraphics[width=\linewidth]{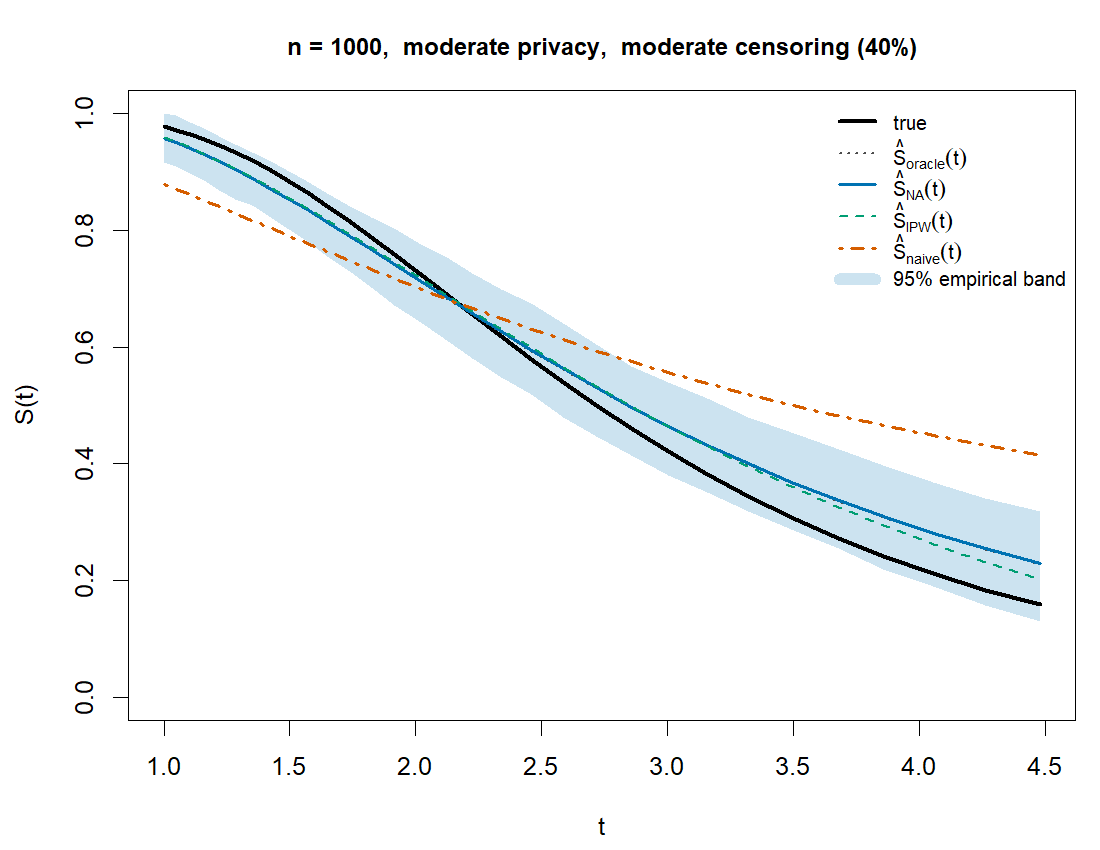}
    \caption{Moderate privacy ($\epsilon_{\mathrm{time}}=4$), moderate censoring}
  \end{subfigure}
  \hfill
  \begin{subfigure}[t]{0.47\textwidth}
    \centering
    \includegraphics[width=\linewidth]{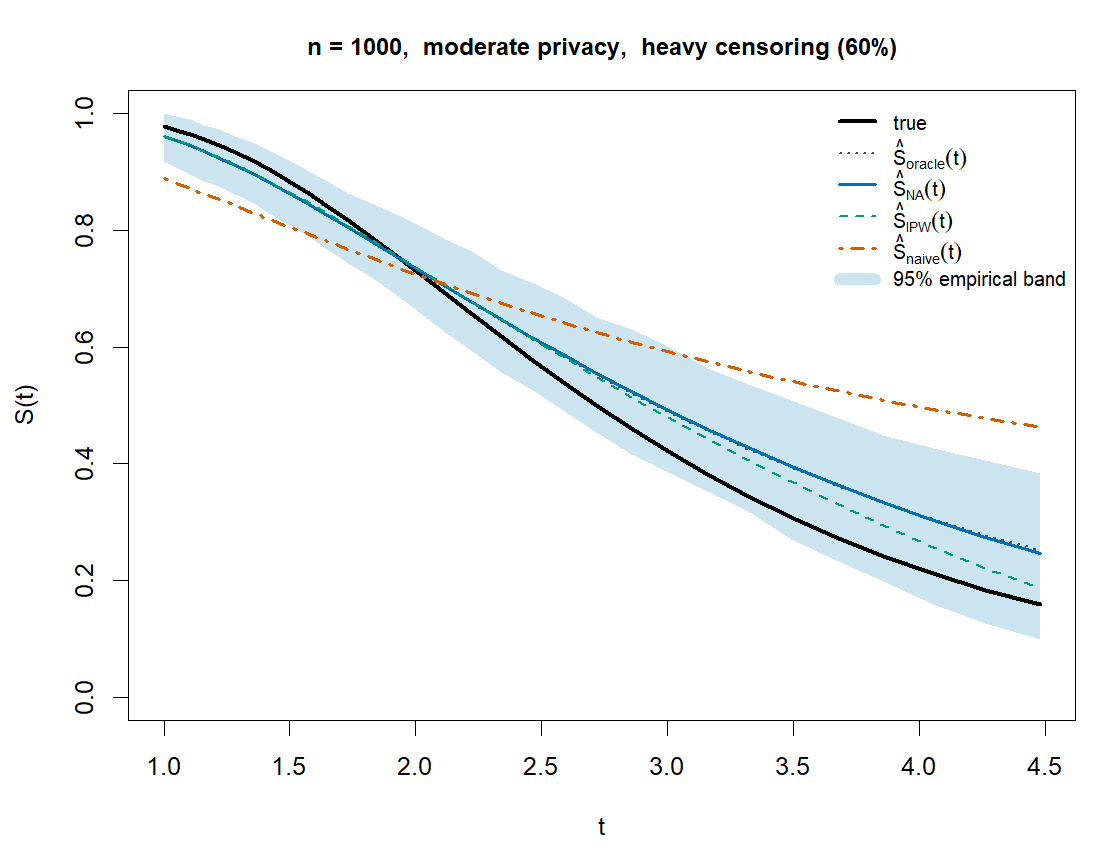}
    \caption{Moderate privacy ($\epsilon_{\mathrm{time}}=4$), heavy censoring}
  \end{subfigure}

  \vspace{0.6em}
  \begin{subfigure}[t]{0.47\textwidth}
    \centering
    \includegraphics[width=\linewidth]{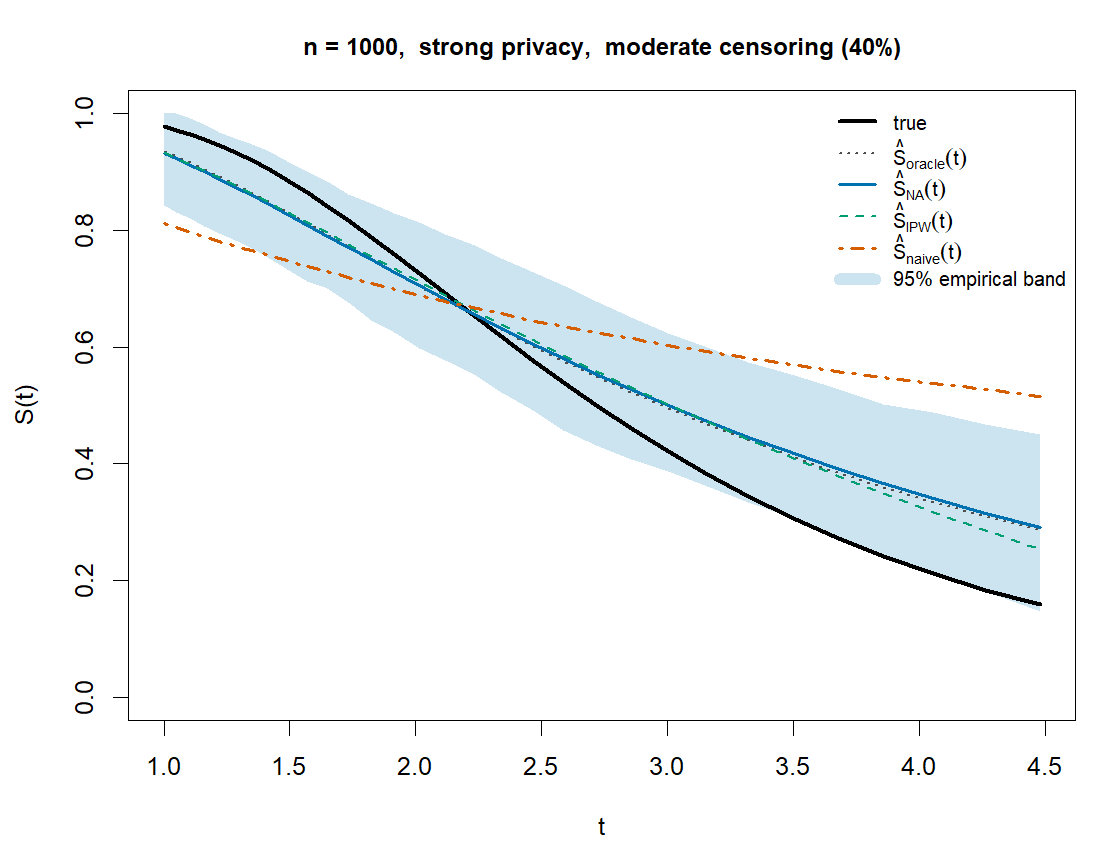}
    \caption{Strong privacy ($\epsilon_{\mathrm{time}}=2$), moderate censoring}
  \end{subfigure}
  \hfill
  \begin{subfigure}[t]{0.47\textwidth}
    \centering
    \includegraphics[width=\linewidth]{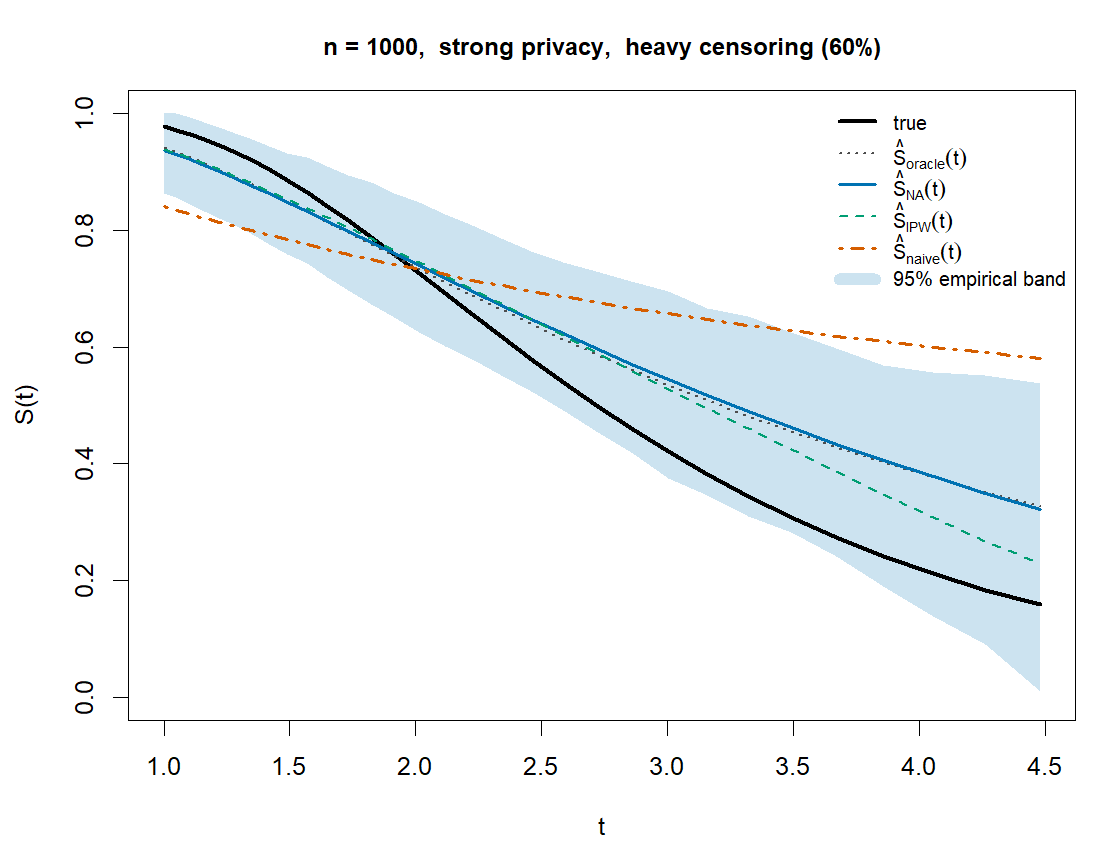}
    \caption{Strong privacy ($\epsilon_{\mathrm{time}}=2$), heavy censoring}
  \end{subfigure}

  \caption{Simulation I: Estimated survivor functions and approximated 95\% confidence bands for $300$ runs, $n=1000$.}
  \label{fig:curves1000}
\end{figure}
 \begin{figure}[htbp]
  \centering
  \begin{subfigure}[t]{0.47\textwidth}
    \centering
    \includegraphics[width=\linewidth]{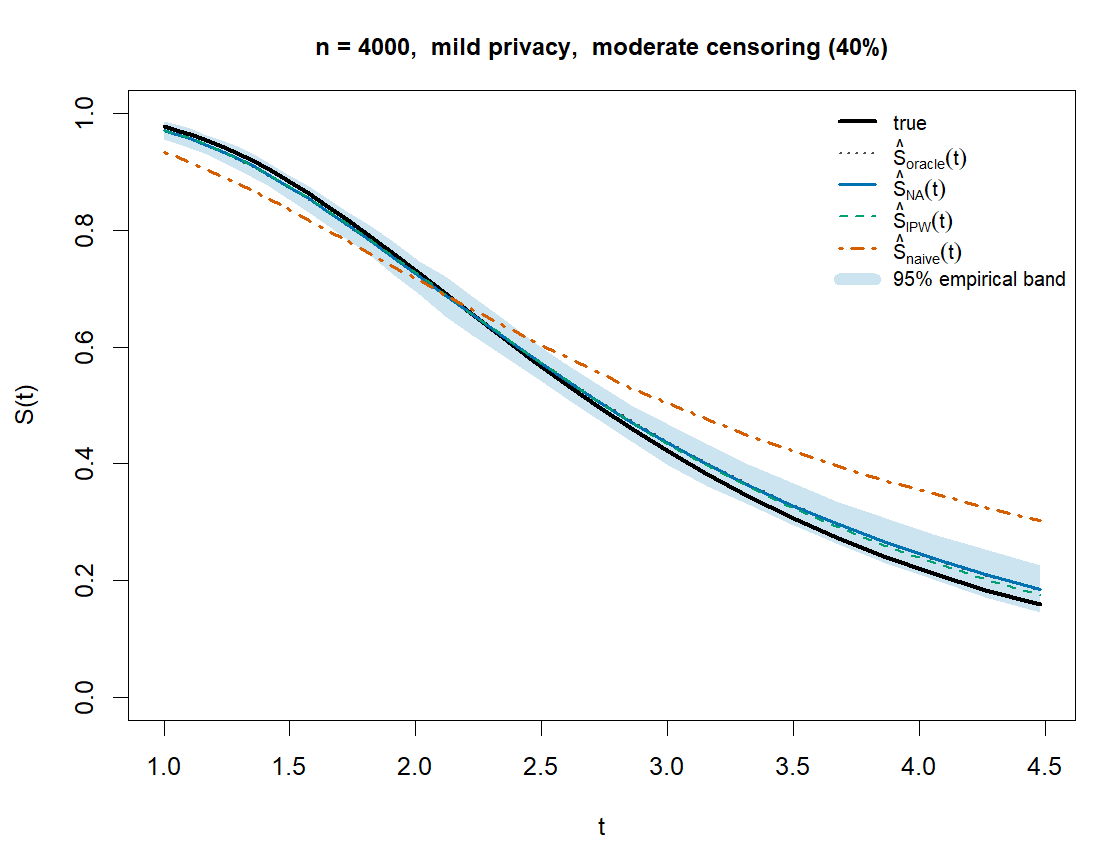}
    \caption{Mild privacy ($\epsilon_{\mathrm{time}}=8$), moderate censoring}
  \end{subfigure}
  \hfill
  \begin{subfigure}[t]{0.47\textwidth}
    \centering
    \includegraphics[width=\linewidth]{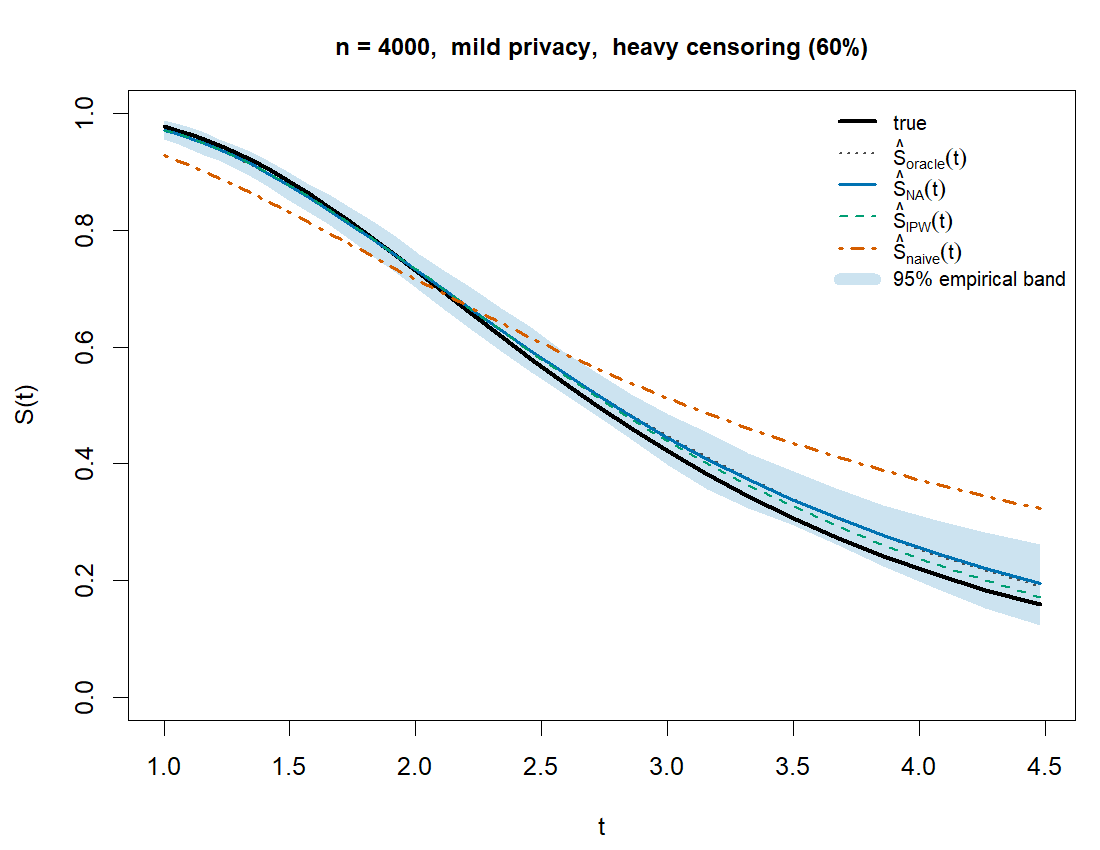}
    \caption{Mild privacy ($\epsilon_{\mathrm{time}}=8$), heavy censoring}
  \end{subfigure}

  \vspace{0.6em}
  \begin{subfigure}[t]{0.47\textwidth}
    \centering
    \includegraphics[width=\linewidth]{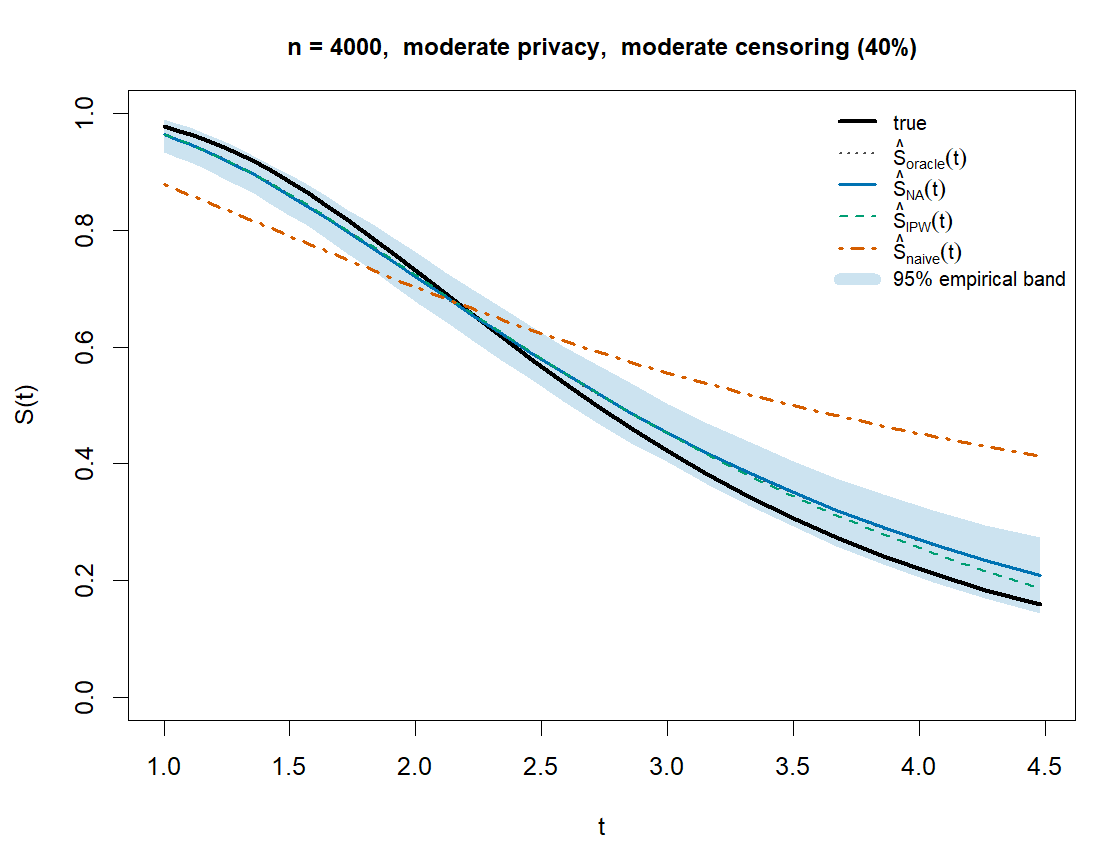}
    \caption{Moderate privacy ($\epsilon_{\mathrm{time}}=4$), moderate censoring}
  \end{subfigure}
  \hfill
  \begin{subfigure}[t]{0.47\textwidth}
    \centering
    \includegraphics[width=\linewidth]{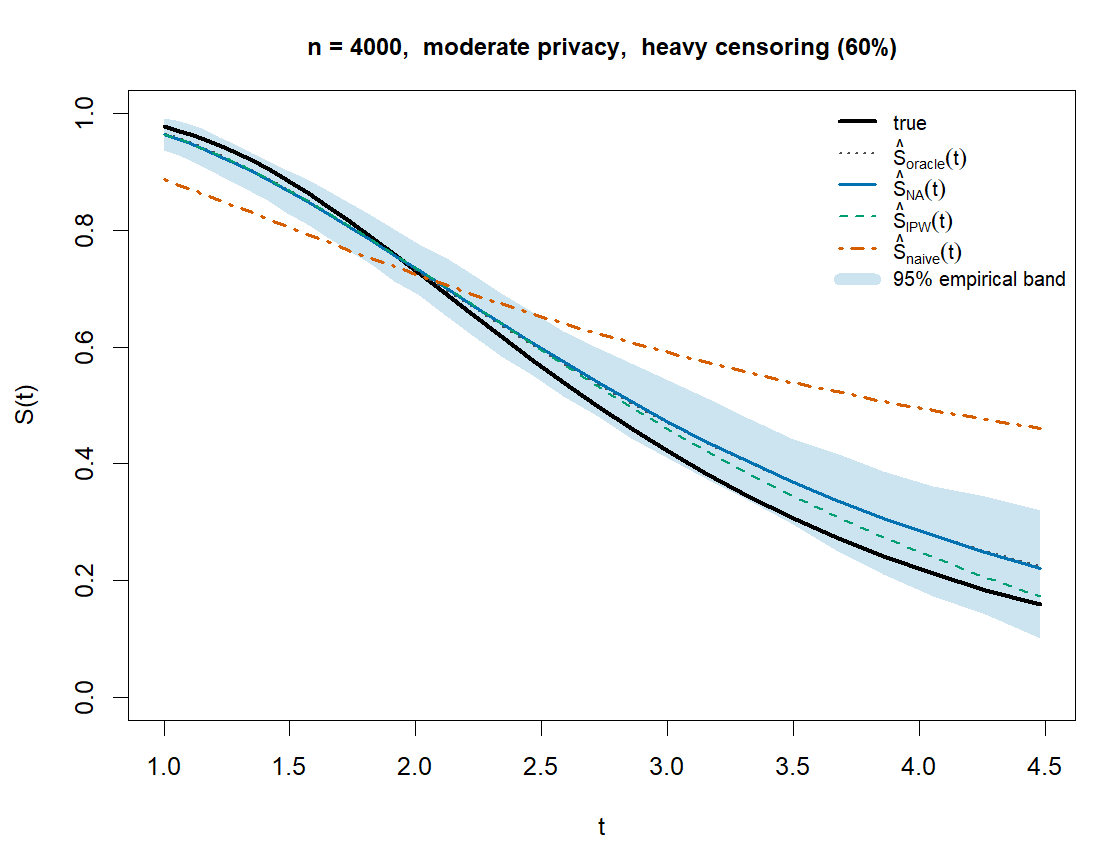}
    \caption{Moderate privacy ($\epsilon_{\mathrm{time}}=4$), heavy censoring}
  \end{subfigure}

  \vspace{0.6em}
  \begin{subfigure}[t]{0.47\textwidth}
    \centering
    \includegraphics[width=\linewidth]{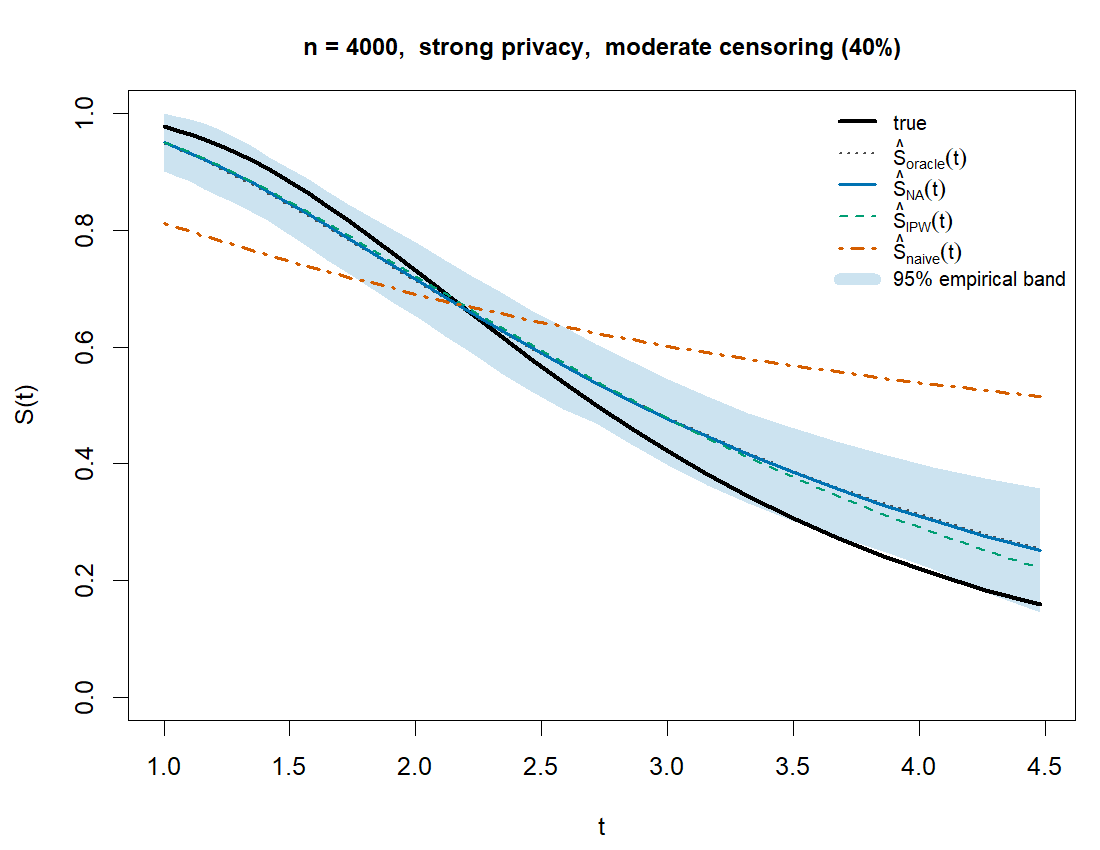}
    \caption{Strong privacy ($\epsilon_{\mathrm{time}}=2$), moderate censoring}
  \end{subfigure}
  \hfill
  \begin{subfigure}[t]{0.47\textwidth}
    \centering
    \includegraphics[width=\linewidth]{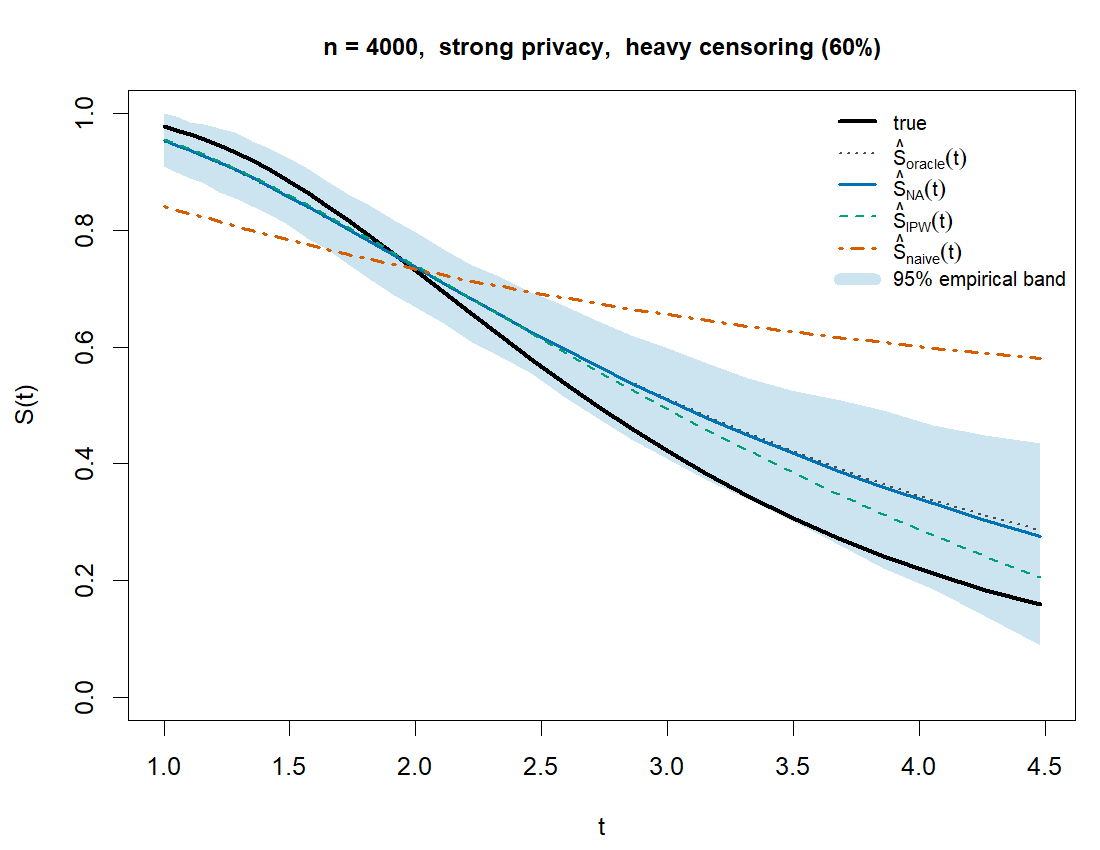}
    \caption{Strong privacy ($\epsilon_{\mathrm{time}}=2$), heavy censoring}
  \end{subfigure}

  \caption{Simulation I: Estimated survivor functions and approximated 95\% confidence bands for $300$ runs, $n=4000$. }
  \label{fig:curves4000}
\end{figure}

\subsection{Simulation II}
In this simulation, we assess the impact of bandwidth on the proposed estimator $\widehat{S}_{\mathrm{NA}}(\cdot)$. We
used the same data generating mechanism and the same three privacy budgets as
in Simulation I and chose the moderate censoring to simulate the data. We evaluated $\widehat{S}_{\mathrm{NA}}(\cdot)$ over a grid of bandwidths. With the true survivor function, $S_T(\cdot)$, we computed MAE and the maximum bias, which is $\sup_{t\in\mathcal{T}}|\widehat{S}_{\mathrm{NA}}(t)-S_{T}(t)|$. We compared the bandwidths selected by minimizing these two criteria with the bandwidth $h=(b^{4}/n)^{1/7}$ recommended in Remark~\ref{rm:bw}.

Table~\ref{tab:sim2} presents the theoretical bandwidth and the sample mean of bandwidths selected to minimize MAE and the supremum bias in each run across 400 repetitions. The empirical bandwidths of the two criteria are close to each
other and both are close to $h_{\mathrm{opt}}$. At $n=1000$ the bandwidth
$h_{\mathrm{opt}}$ coincides exactly with the sample mean of the empirical bandwidths under all three privacy budgets. At $n=4000$ the empirical bandwidth selected by MAE is slightly below $h_{\mathrm{opt}}$
while the one selected via the supremum bias agrees with $h_{\mathrm{opt}}$ in two of the three scenarios. Figure \ref{fig:fig3bw} plots the sample means of MAE and the supremum bias on a grid of bandwidths and compares it with the selected bandwidths. The agreement of two sets of bandwidths confirms the optimal bandwidth selected by Theorem~\ref{thm:main}.

\begin{table}[htbp]\centering\small
\caption{Simulation II: theoretically recommended bandwidth $h_{\mathrm{opt}}=(b^{4}/n)^{1/7}$ and the sample mean of empirical bandwidths that minimize MAE and sup-bias over 400 runs}
\label{tab:sim2}
\begin{tabular}{ll c cc}\toprule
& &{theoretical bandwidth} & \multicolumn{2}{c}{mean empirical bandwidths} \\
\cline{4-5}\\
$n$ & privacy & $h_{\mathrm{opt}}$ & MAE&sup-bias \\ \midrule
1000 & strong & 0.461  & 0.461 & 0.461\\
 & moderate & 0.310& 0.310 & 0.310 \\
 & mild & 0.209  & 0.209 & 0.209\\
\midrule
4000 & strong & 0.378 & 0.318 & 0.378 \\
 & moderate & 0.254& 0.214 & 0.214 \\
 & mild & 0.171 & 0.144 & 0.171\\
\bottomrule\end{tabular}\end{table}

\begin{figure}[htbp]
    \centering
    \begin{subfigure}[b]{0.32\textwidth}
        \centering
        \includegraphics[width=\textwidth]{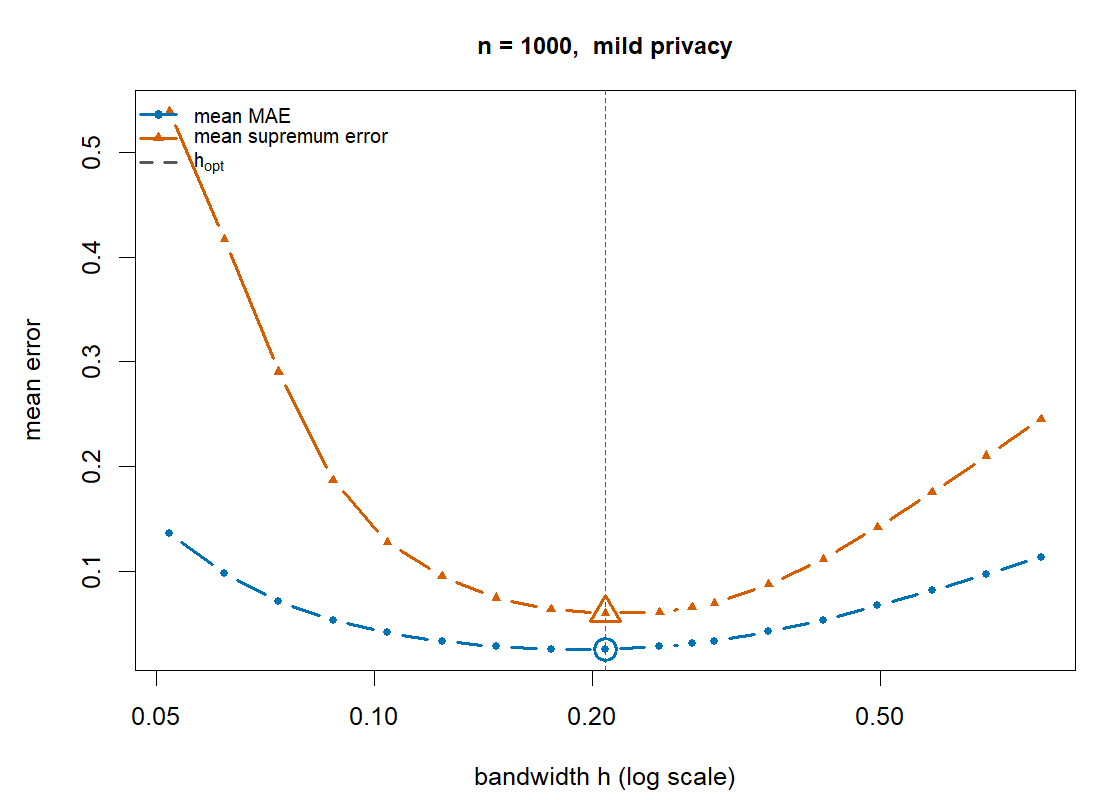}
        \caption{Mild privacy, $n=1000$}
        \label{fig:sub1}
    \end{subfigure}
    \hfill
    \begin{subfigure}[b]{0.32\textwidth}
        \centering
        \includegraphics[width=\textwidth]{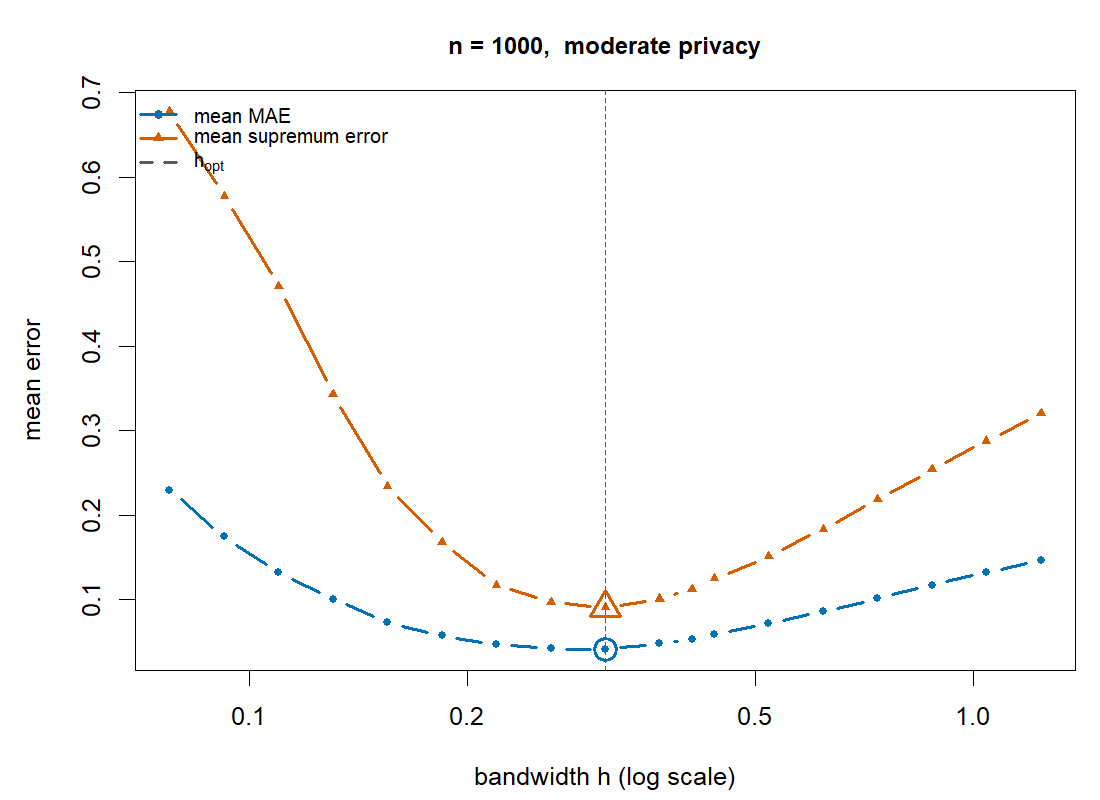}
        \caption{Moderate privacy, $n=1000$}
        \label{fig:sub2}
    \end{subfigure}
    \hfill
    \begin{subfigure}[b]{0.32\textwidth}
        \centering
        \includegraphics[width=\textwidth]{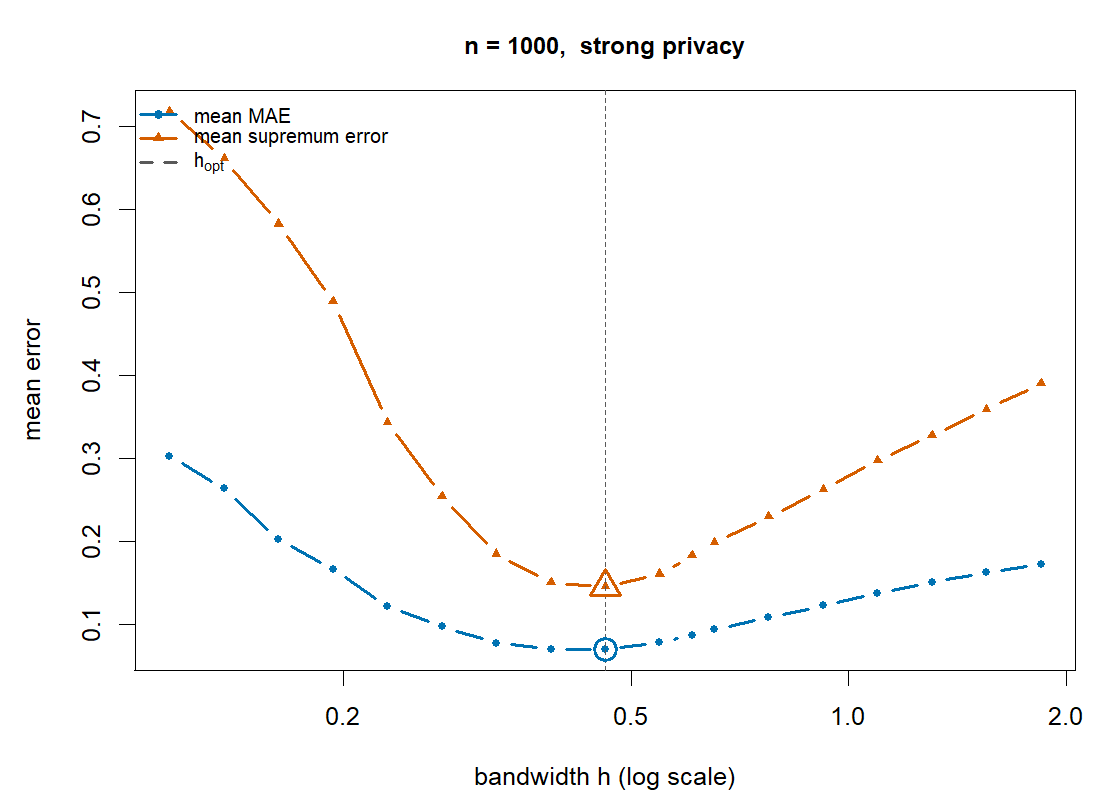}
        \caption{Strong privacy, $n=1000$}
        \label{fig:sub3}
    \end{subfigure}
    
    \vspace{1em} 
    \begin{subfigure}[b]{0.32\textwidth}
        \centering
        \includegraphics[width=\textwidth]{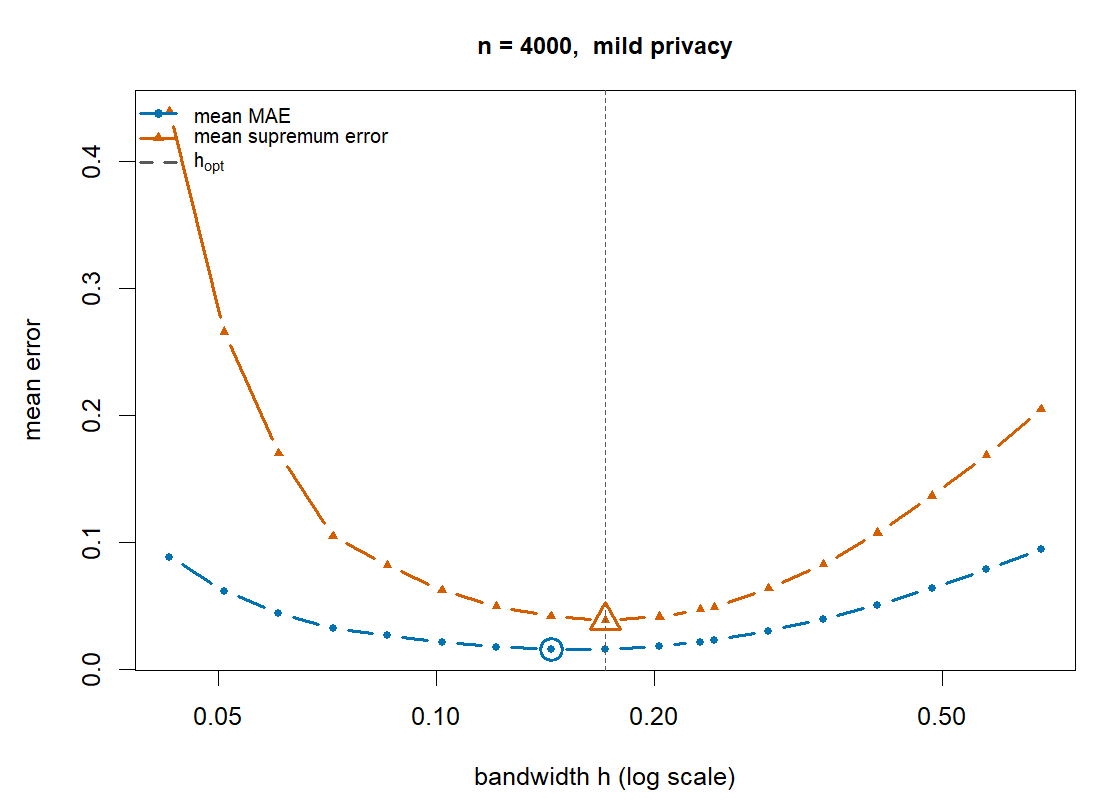}
        \caption{Mild privacy, $n=4000$}
        \label{fig:sub4}
    \end{subfigure}
    \hfill
    \begin{subfigure}[b]{0.32\textwidth}
        \centering
        \includegraphics[width=\textwidth]{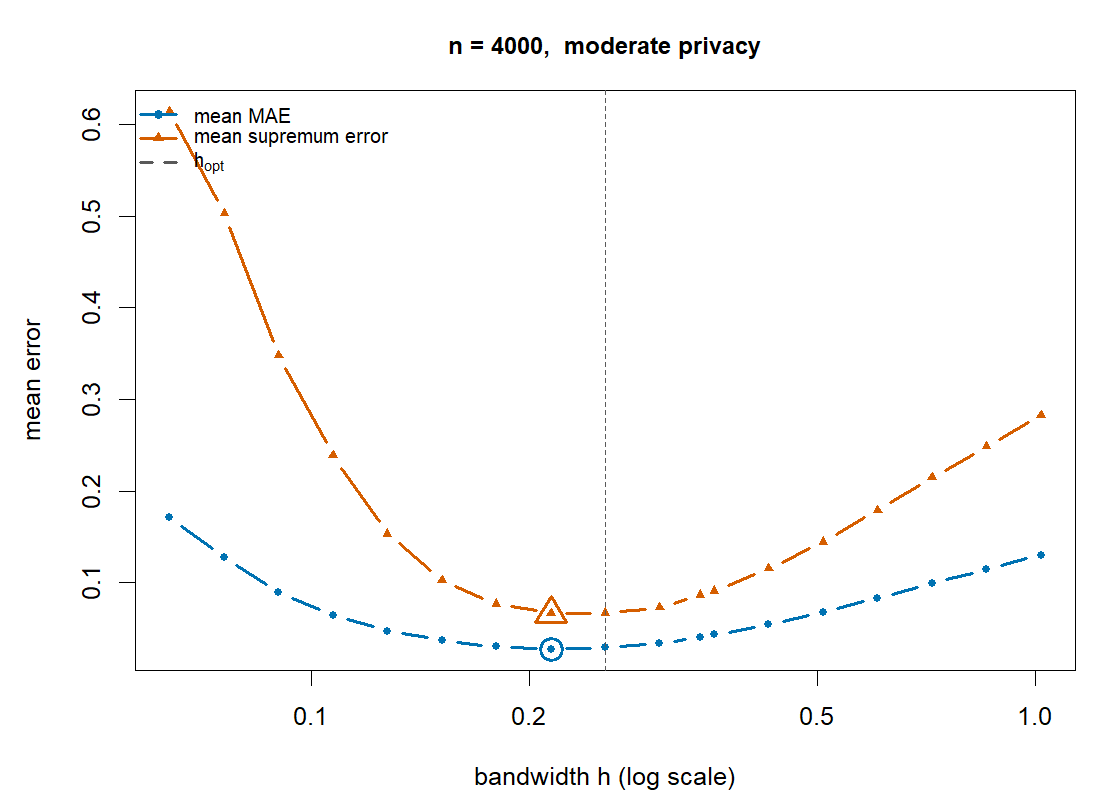}
        \caption{Moderate privacy, $n=4000$}
        \label{fig:sub5}
    \end{subfigure}
    \hfill
    \begin{subfigure}[b]{0.32\textwidth}
        \centering
        \includegraphics[width=\textwidth]{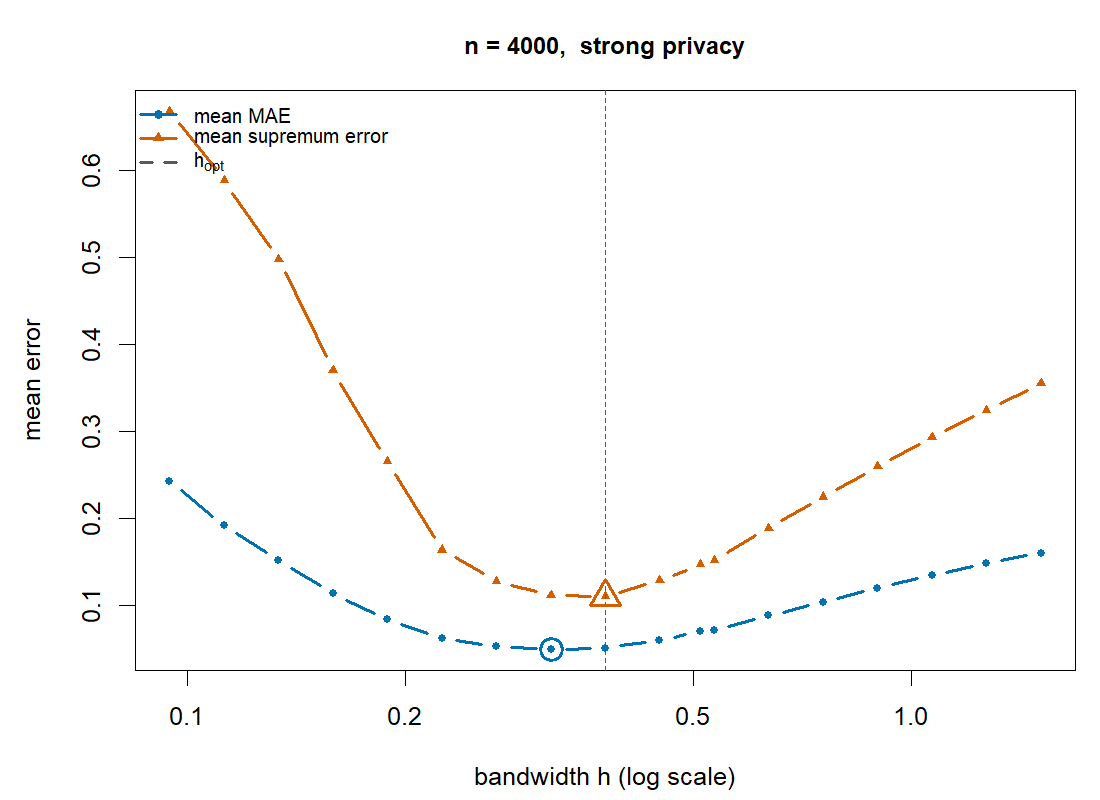}
        \caption{Strong privacy, $n=4000$}
        \label{fig:sub6}
    \end{subfigure}

    \caption{Simulation II: Theoretical bandwidths (the vertical dashed line) and empirical bandwidths}
    \label{fig:fig3bw}
\end{figure}
\section{Application to the METABRIC Data} \label{sec:rda}
We illustrate the proposed method on a semi-synthetic example built from a publicly available breast cancer dataset for Molecular Taxonomy of Breast Cancer International Consortium (METABRIC) cohort, {which was a Canada-UK Project containing targeted sequencing data of primary breast cancer samples and was collected by Cambridge Research Institute and the British Columbia Cancer Centre in Canada (\citealp{Pereira2016}).} We obtained the $1390$ patients with tumour
stage I, II and III. These three subgroups comprise $475$, $800$ and
$115$. The event of
interest is death from any cause. Given follow-up ranges
from $0.1$ to $351$ months, the event times are subject to right censoring with with censoring rates $54.7\%$, $39.8\%$ and $25.2\%$ for three subgroups
respectively and an overall censoring rate is $42\%$, which is close to the scenario in our simulation study with a moderate censoring rate. 

Because the METABRIC data are not privatized per se, we construct a synthetic version of the dataset by adopting LDP mechanism following \eqref{eq:perturb-mech}. For each patient, we consider the observed event times within $[\exp(2.4),\exp(5.4)]$. We computed $\Delta=3.0$, which is the width of the domain and the maximum difference in the records of log-transformed times. We then added
$\mathrm{Laplace}(0,b)$ perturbation to each log-transformed observed times with $b=\Delta/\epsilon_{\mathrm{time}}$. We set
$\epsilon_{\mathrm{time}}=8$ and $2$ for mild and strong privacy budget and thus the corresponding values of parameter $b$ are $b=\Delta/\epsilon_{\mathrm{time}}=0.38$ and $b=1.51$. The event indicator
was privatized by symmetric randomized response with $p_{1}=0.881$ and $p_{0}=0.119$, a
value within the range used by Google's RAPPOR system. The event indicators were flipped for about $11\%$ patients. 

We repeated such privatization procedure for $50$ times and computed four estimators of survivor functions for each tumor stage. Since the non-private data is available to us, we can evaluate two gold standards: the Kaplan--Meier estimator, reported with its Greenwood confidence band,
and the oracle kernel estimator $\widehat{S}_{\mathrm{oracle}}(\cdot)$ defined in
(\ref{S_oracle}) with bandwidth selected using Theorem~\ref{thm:main}. The remaining two estimators use the synthetic LDP data: the
proposed $\widehat{S}_{\mathrm{NA}}(\cdot)$ of \eqref{eq:SNAhat} with point-wise $95\%$ confidence intervals from
the multiplier bootstrap of Theorem~\ref{thm:boot} with $B=400$ draws, and the naive estimator, which is the Kaplan--Meier estimator that does not correct bias from the synthetic dataset. 

We display the four estimated survivor functions in Figures~\ref{fig:metabricmild} and
\ref{fig:metabricstrong}. Both the proposed estimator, $\hat{S}_{\mathrm{NA}}(\cdot)$ and the naive estimator, are close to the two gold standards under mild privacy. However, the two estimators are clearly separated when the privacy restriction gets tight. The naive estimator gets flatten and deviates from the two gold standards across all three tumor stages, whereas the proposed estimator closely follow the shape of the Kaplan--Meier curve with the true dataset. Table \ref{tab:metabricmae} summarizes the mean MAEs of the proposed estimator and the naive estimators over 50 privacy-protected released data. With both mild and strong privacy scenarios, the proposed estimator has much smaller MAEs than the naive estimator, particularly with a strong privacy budget. Our de-biased procedure therefore matters most precisely where local
differential privacy is deployed at the tight budgets. We see that the confidence band with proposed estimator gets wider when privacy tightens. This agrees with the results in simulation study and indicates that the cost of the de-biasing procedure is in variance estimation. We also compute the MAE for an oracle kernel estimator. The minor MAEs of the oracle estimator indicated the validity of the bandwidth given in Theorem~\ref{thm:main} to balance the smoothing bias and the convergence error. 
\begin{figure}[htbp]
  \centering
  \setlength{\tabcolsep}{2pt}
  \begin{tabular}{@{}ccc@{}}
    \includegraphics[width=0.32\textwidth]{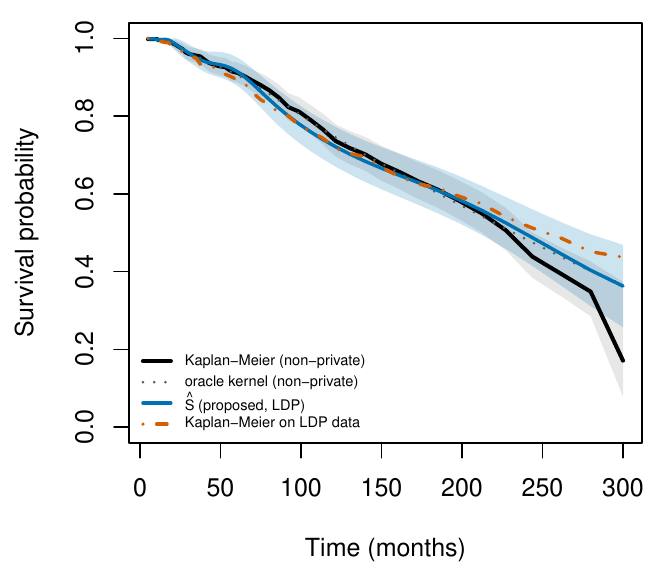} &
    \includegraphics[width=0.32\textwidth]{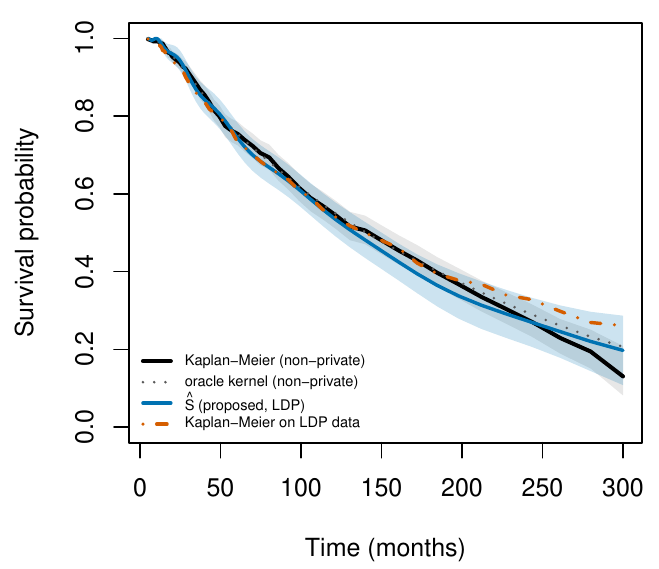} &
    \includegraphics[width=0.32\textwidth]{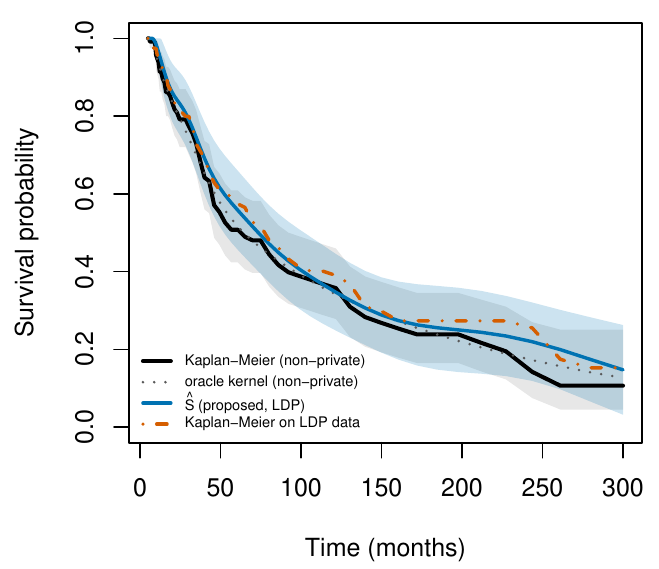}\\[-2pt]
    {\small (a) Stage I ($n=475$)} &
    {\small (b) Stage II ($n=800$)} &
    {\small (c) Stage III ($n=115$)}
  \end{tabular}
  \caption{METABRIC data under \emph{mild} privacy
  ($\epsilon_{\mathrm{time}}=8$).}
  \label{fig:metabricmild}
\end{figure}

\begin{figure}[htbp]
  \centering
  \setlength{\tabcolsep}{2pt}
  \begin{tabular}{@{}ccc@{}}
    \includegraphics[width=0.32\textwidth]{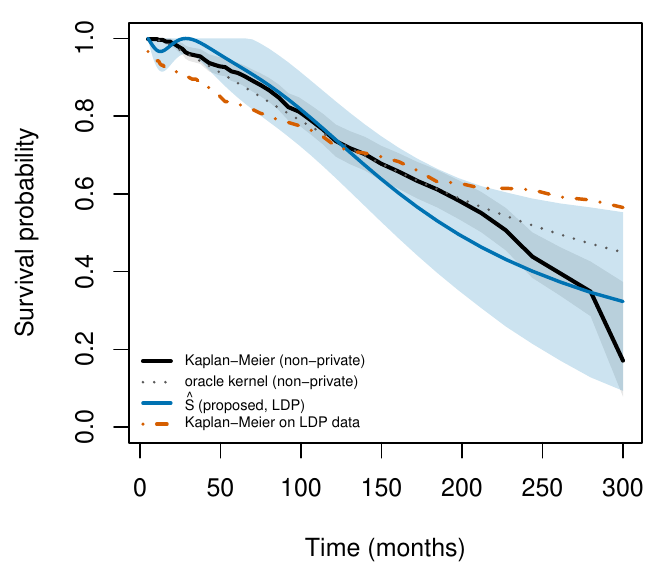} &
    \includegraphics[width=0.32\textwidth]{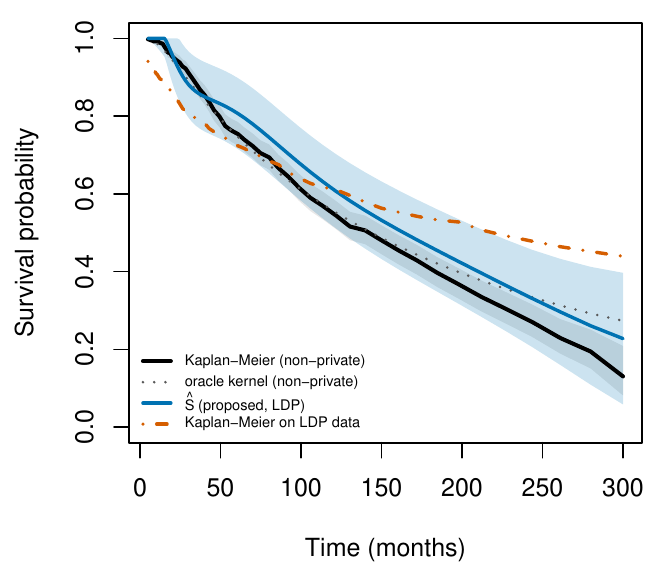} &
    \includegraphics[width=0.32\textwidth]{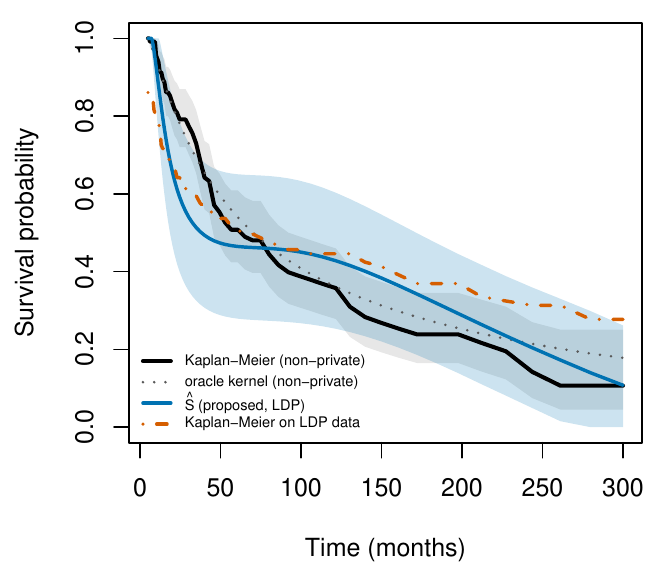}\\[-2pt]
    {\small (a) Stage I ($n=475$)} &
    {\small (b) Stage II ($n=800$)} &
    {\small (c) Stage III ($n=115$)}
  \end{tabular}
  \caption{METABRIC data under \emph{strong} privacy
  ($\epsilon_{\mathrm{time}}=2$)}
  \label{fig:metabricstrong}
\end{figure}
\begin{table}[htbp]
  \centering\small
  \caption{Mean and standard deviation of the MAE for the proposed estimator $\hat{S}_{\mathrm{NA}}(\cdot)$ and the na\"iv estimator over $50$ independent
  locally private releases of the METABRIC data }
  \label{tab:metabricmae}
  \begin{tabular}{llcccc}
  \toprule
  & &{$\widehat{S}_{\mathrm{oracle}}(\cdot)$} &  $\widehat{S}_{\mathrm{NA}}(\cdot)$
    &$\widehat{S}_{\mathrm{naive}}(\cdot)$ & \\
  \cmidrule(lr){3-3}\cmidrule(lr){4-4}\cmidrule(lr){5-5}
  Setting & Stage &  & mean (SD) & mean (SD)\\
  \midrule
   Mild
          & Stage I   & 0.0083 & 0.0136 (0.0025) & 0.0184 (0.0033) \\
  ($b=0.378$)
          & Stage II  & 0.0061 & 0.0120 (0.0022) & 0.0180 (0.0028)\\
          & Stage III & 0.0127 & 0.0265 (0.0064) & 0.0265 (0.0053) \\
          
  \addlinespace
 
  Strong
          & Stage I   & 0.0181 & 0.0387 (0.0125) & 0.0770 (0.0077)\\
  ($b=1.510$)
          & Stage II  & 0.0164 & 0.0379 (0.0090) & 0.0917 (0.0055) \\
          & Stage III & 0.0257 & 0.0745 (0.0363) & 0.1035 (0.0144)\\
  \bottomrule
  \end{tabular}
\end{table}
\section{Discussion and Conclusion} \label{sec:summary}
Protecting individual-level privacy inevitably trades privacy against
statistical utility. This is particularly strong for time-to-event data under local differential privacy to gain strong protection because each time record is perturbed with noise and each event indicator has also been flipped with a known probability. Applying conventional survival analysis methods to such privatized data without explicit debiasing leads to severe inferential bias and their bias does not vanish as the sample size grows.

This paper proposes a novel nonparametric deconvolution estimator for the survivor function with right-censored records released under Local Differential Privacy (LDP) to address the issue of inferential bias with privatized data. Our estimator makes two contributions. First, the estimator reconciles privacy with valid statistical inference and it corrects the two-layer biases at the level of the distribution via deconvolution kernels removing the biases on the distribution of event times, and moment reconstruction removing the bias from the flipped indicators. Our uniform consistent estimator is inspired by the Nelson--Aalen estimator so it requires neither knowledge nor estimation of the censoring distribution. Because it uses only the released data, it is a post-processing of the privatized records and inherits the privacy guarantee imposed by LDP data sharing mechanism. With our proposed estimator, privacy is guaranteed at data collection, while the analyst still obtains a consistent estimator with quantified uncertainty. Second, the derived uniform convergence rate provides a protocol for balancing privacy and
utility. The uniform convergence rate depends on the sample size and the privacy
budget, so the cost of privacy at the growth of sample size can be stated explicitly. Under the Laplace mechanism, halving the privacy budget requires
about sixteen times as many records to maintain the same accuracy. Data holders
can thus weigh a stronger privacy guarantee against the additional data needed
to preserve utility, and plan the privacy budget and sample size jointly when designing a mechanism for data sharing.

Several limitations point to future work. First, the method assumes that the
perturbation distribution is known because the releasing mechanism is usually published. Although our implementation uses the closed-form deconvolution
kernel for Laplace perturbation, the procedure and theory apply to any ordinarily
smooth perturbation, for which only the deconvolution kernel $\mathcal{K}_{h}(\cdot)$ needs to be recomputed. Second, our method relies on the independence assumption of censoring times and event times. While this assumption routinely holds under administrative censoring, future research should generalize this framework to more complex survival structures, such as dependent censoring and competing risks. Finally, extending our debiasing framework to regression analysis of time-to-event data would be a potential avenue for future research to allow assessing risk with data under LDP.

	\section*{Acknowledgments}
This research was supported
by grants from the National Science and Technology Council (114-2118-M-004-010-MY3) and the National Science Foundation (2144684). Any opinions, findings, and conclusions or recommendations expressed in this material are those of the author(s) and do not necessarily reflect the views of the  National Science and Technology Council and the National Science Foundation.

\bibliography{references}


\clearpage
\setcounter{footnote}{1}
\setcounter{page}{1}   

\baselineskip=28pt
\begin{center}
{\LARGE{\bf Supplementary Materials of ``De-biased Distribution Estimation for Locally Differential Private Survival Data''}}
\end{center}
\baselineskip=14pt
\vskip 2mm

\begin{center}
Yi Xiong$^1$, Li-Pang Chen$^2$, and Liyue Fan$^3$\\
\ \\
$^1$Department of Biostatistics, The State University of New York (SUNY) at Buffalo\\
$^2$Department of Statistics, National Chengchi University\\
$^3$Department of Computer Science, University of North Carolina at Charlotte
\end{center}

This supplementary material outlines our derivation for the asymptotic properties of the proposed estimator for Theorems 1 and 2 of Section 2.3.
\section*{A. Proof of Theorem 1}
\begin{proof}
We establish the uniform convergence rate of $\widehat{S}_{\mathrm{NA}}(t)$ over $\mathcal{T}$ by linearizing the functional mapping, evaluating the smoothing bias and applying empirical process bounds to the stochastic fluctuations. Consider $t \in \mathcal{T}$. We write $\ell_0 := \log t_0'$ and $\kappa := \inf_{u \le \log \tau'} H_Y(u) = H_Y(\log \tau') \ge \eta > 0$ by Assumption~\ref{as:trunc}. Let $C^*, C^*_1, C^*_2, \dots$ represent the finite and positive constants that are independent of $n$, $h$ and $b$.

\medskip

We observed that $\widehat{S}_{\mathrm{NA}}(t) = \exp\{ -\widehat{\Lambda}(\log t) \}$ is a Nelson--Aalen-type estimator, where $\widehat{\Lambda}(\log t) = \int_{\ell_0}^{\log t} \frac{d\widehat{H}_1(u)}{\widehat{H}_Y(u)}$. By Assumption~\ref{as:trunc}, the mapping $(H_1, H_Y) \mapsto \exp\left(-\int_{\ell_0}^{\cdot} H_Y^{-1} dH_1\right)$ is Hadamard differentiable on the space of bounded variation functions bounded below by $\inf_{u} H_Y(u) \ge \kappa > 0$, and a first-order functional Taylor's expansion yields
\begin{equation}
  \widehat{S}_{\mathrm{NA}}(t) - S_T(t) = - S_T(t) \int_{\ell_0}^{\log t} \frac{d\widehat{H}_1(u) - dH_1(u)}{H_Y(u)} + S_T(t) \int_{\ell_0}^{\log t} \frac{\widehat{H}_Y(u) - H_Y(u)}{H_Y^2(u)} dH_1(u) + R_n(t),
  \label{eq:full_hadamard_proof}
\end{equation}
where $d\widehat{H}_1(u) = \widehat{g}_1(u) du$, $H_1(u) = g_1(u)du$, and $R_n(t)$ is the remainder. Note the left hand side of \eqref{eq:full_hadamard_proof} can be written as
\[
 \widehat{S}_{\mathrm{NA}}(t) - S_T(t) = [\widehat{S}_{\mathrm{NA}}(t) - \mathbb{E}\widehat{S}_{\mathrm{NA}}(t)] + [\mathbb{E}\widehat{S}_{\mathrm{NA}}(t) - S_T(t) ],
\]
where the first term is a stochastic error and the second is the smoothing bias. We will provide the bound for each of them.

\noindent{\bf The stochastic error}
The stochastic term is
\begin{equation}
\label{eq:stochastic1}
  \widehat{S}_{\mathrm{NA}}(t) - \mathbb{E}\widehat{S}_{\mathrm{NA}}(t) = - S_T(t) \int_{\ell_0}^{\log t} \frac{\widehat{g}_1(u) - \mathbb{E}\widehat{g}_1(u)}{H_Y(u)} du + S_T(t) \int_{\ell_0}^{\log t} \frac{\widehat{H}_Y(u) - \mathbb{E}\widehat{H}_Y(u)}{H_Y^2(u)} g_1(u) du + \widetilde{R}_n(t),
\end{equation}
where $\widetilde{R}_n(t) = R_n(t) - \mathbb{E}R_n(t)$. With the generic notation of one realization $o^* = (y^*, \widehat{\delta})$, let $\xi(t; o^*) = \int_{\ell_0}^{\log t} \mathbf{W}_t(u)^\top \boldsymbol{\psi}(u; o^*)\, du$, where $\boldsymbol{\psi}(u; o^*) = \begin{pmatrix} \mathcal{K}_h(u - y^*) \cdot \widehat{\delta} \\ \int_u^\infty \mathcal{K}_h(v - y^*) \, dv\end{pmatrix}$, $\mathbf{W}_t(u) \triangleq \begin{pmatrix} -\frac{S_T(t)}{H_Y(u)} \\ \frac{S_T(t)g_1(u)}{H_Y^2(u)} \end{pmatrix}=\begin{pmatrix} W_{t,1}(u)\\ W_{t,2}(u)\end{pmatrix}$. Further, let $\mathrm{P}_n$ and $\mathrm{P}$ be the empirical measure and probability measure of $n$ i.i.d. observations $\mathcal{O}^*_i=(Y^*_i,\widehat{\delta}_i)$ for $i=1,\cdots,n$. We can present \eqref{eq:stochastic1} as
\begin{align}
\label{eq: stochastic2}
& \frac{1}{n}\sum_{i=1}^n \int_{\ell_0}^{\log t} \mathbf{W}_t(u)^\top \boldsymbol{\psi}(u;\mathcal{O}^*_i) \, du-\int_{\ell_0}^{\log t} \mathbf{W}_t(u)^\top \mathbb{E}[\boldsymbol{\psi}(u; \mathcal{O}^*_i)] du+\widetilde{R}_n(t) \nonumber \\
=\;&\mathrm{P}_n\,\xi(t;\cdot)-\mathrm{P}\,\xi(t;\cdot)+\widetilde{R}_n(t) \nonumber \\
=\;&\Big(\mathrm{P}_n-\mathrm{P}\Big)\xi(t;\cdot)+\widetilde{R}_n(t) .
\end{align}
Define the function class $\mathcal{F}_{n}=\{o^*\mapsto\xi(t;o^*):t\in\mathcal{T}\}$, and write $\|\mathrm{P}_n-\mathrm{P}\|_{\mathcal{F}_n}=\sup_{t\in\mathcal{T}}|(\mathrm{P}_n-\mathrm{P})\xi(t;\cdot)|$. Since $\xi(t;\cdot)$ depends on $h$ and $h$ shrinks according to $n$, we incorporate $n$ in the notation of function class. We will show $\mathcal{F}_{n}$ is a class of Vapnik--Chervonenkis(VC) type and will apply Talagrand's inequality in the form given by \citet[Theorem~2.3]{gine2002rates} to derive the bound. To show that $\mathcal{F}_{n}$ is a VC class, we will establish three uniform bounds: the envelope, the covering number, and the variance of $\xi(t;\mathcal{O}^*_i)$.

Let $\overline{\mathcal{K}}_{h}(z)=\int_{-\infty}^{z} \mathcal{K}_{h}(v)\,dv$. For Laplace perturbation with scale parameter $b$, the characteristic function is $\varphi_e(w) = (1 + b^2 w^2)^{-1}$, which satisfies Assumption~\ref{as:noise} with smoothness index $\beta = 2$ and $c_e = b^{-2}$. Hence, under Assumption~\ref{as:kernel}(i) and (iii), $\overline{\mathcal{K}}_h(z) = \overline{K}\left(\frac{z}{h}\right) - \frac{b^2}{h^2} K'\left(\frac{z}{h}\right)$, where $\overline{K}(w) = \int_{-\infty}^w K(v) dv$. We then have $\|\overline{\mathcal{K}}_h\|_\infty \le C^*_1 \left( 1 + \frac{b^2}{h^2} \|K'\|_\infty \right) \le C^*_1 b^2 h^{-2}$ with constant $C^*_1$, for $h$ small.

Recall the individual functional:
\[
\xi(t; o^*) = \int_{\ell_0}^{\log t} \left[ W_{t,1}(u) \mathcal{K}_h(u - y^*) \widehat{\delta} + W_{t,2}(u) \int_u^{\infty} \mathcal{K}_h(v - y^*) \, dv \right] du.
\]
Because $W_{t,1}, W_{t,2}$ and the derivative $W_{t,1}'$ are uniformly bounded on $\mathcal{T}$ by Assumptions~\ref{as:trunc} and \ref{as:smooth}, $\widehat{\delta}$ is bounded by Assumption~\ref{as:ancillary}, and $\overline{K}, K'$ are bounded functions,
the envelope of $\mathcal{F}_n$ satisfies
\begin{equation}
F := \sup_{t \in \mathcal{T}} \|\xi(t; \cdot)\|_\infty \le C^*_1 \|\overline{\mathcal{K}}_h\|_\infty \le C^*_1 b^2 h^{-2}.
\label{eq:envelope_laplace}
\end{equation}

Further, the class $\mathcal{F}_{n}$ is indexed by the single
real parameter $t$ in a compact interval $\mathcal{T}$ and
$t\mapsto\xi(t;o^*)$ is Lipschitz uniformly in $o^*$ because its first derivative is bounded, that is,
\[
L := \sup_{o^*}\sup_{t \in \mathcal{T}} \left| \frac{\partial \xi(t; o^*)}{\partial t}\right| \le C^*_3 \left( \|\overline{\mathcal{K}}_h\|_\infty + \|\mathcal{K}_h\|_\infty \right) \le C^*_3 b^2 h^{-3}
\]
is finite. By this Lipschitz property, an $\nu/L$-net over $\mathcal{T}$ induces an $\nu$-cover of $\mathcal{F}_n$ under the supremum norm $\Vert{}\cdot\Vert{}_\infty$. Since the supremum norm dominates the $L_2(Q)$ norm for any probability measure $Q$, this net also forms an $\nu$-cover of $\mathcal{F}_n$ in $L_2(Q)$. The cardinality of such a net requires at most $\vert{}\mathcal{T}\vert{}L/\nu$ elements. Following the presentation in \citet{gine2002rates}, we define $N\left(\mathcal{F}_n, L_2(Q), \nu F \right)$ as the $\nu$-covering number of the metric space $(\mathcal{F}_n, L_2(Q))$. Given the cardinality we derived for $\nu$-cover of $\mathcal{F}_n$ and the envelope's bound derived in \eqref{eq:envelope_laplace}, we have
\begin{equation}
\label{eq:covernumber}
\sup_{Q} N\left(\mathcal{F}_n, L_2(Q), \nu F \right) \le \frac{ \vert{}\mathcal{T}\vert{} L}{\nu F}  = \frac{A_n}{\nu}, \qquad A_n := \frac{C^*_3\vert{}\mathcal{T}\vert{}}{C^*_1}h^{-1},
\end{equation}
so that $\mathcal{F}_n$ is of VC type with index $v=1$.

To apply Talagrand's inequality, we also need to derive the uniform bound for the variance of $\xi(t;\mathcal{O}^*_i)$. We observe that $\frac{b^2}{h^2} K'\left(\frac{z}{h}\right)$ in $\overline{\mathcal{K}}_h(z)$ is localized on a window of width $O(h)$ about $z = 0$, whereas $\overline{K}(z/h)$ is bounded by one. Since the density of $Y^*$ is bounded by Assumption~\ref{as:smooth}, integrating the squared envelope over this $h$-width window gains one power of $h$, so that
\begin{equation}
\sigma_n^2 := \sup_{t \in \mathcal{T}} \mathrm{P}\,\xi(t; \cdot)^2 \le C^*_4 \left\{1+ \left( b^2 h^{-2} \right)^2 \cdot h\right\} \le C^*_4 b^4 h^{-3},
\label{eq:variance_laplace}
\end{equation}
the last step because $b^4h^{-3}\to\infty$ under Assumption~\ref{as:bw}; note also $\sigma_n\le F$, as required. With the uniform bounds given in envelope function, covering number and weak variance, \citet[Theorem~2.3]{gine2002rates} gives
\begin{equation}
  \E\|\mathrm{P}_{n}-\mathrm{P}\|_{\mathcal{F}_{n}}
  \le C^*\left\{\left(\frac{\sigma_n^{2}\log(A_nF/\sigma_n)}{n}\right)^{1/2}
  +\frac{F\log(A_nF/\sigma_n)}{n}\right\}.
  \label{eq:talagrand}
\end{equation}
By \eqref{eq:envelope_laplace}, \eqref{eq:covernumber} and \eqref{eq:variance_laplace}, $\frac{A_nF}{\sigma_n} \le \frac{C^{*}_5 h^{-3}}{\sqrt{C^{*}_4} b^2 h^{-3/2}} = C^{*}_6 b^{-2}h^{-3/2}$. Therefore, $\log\left( \frac{A_nF}{\sigma_n} \right) = O(\log(1/h)) = O(\log n)$. The two terms of \eqref{eq:talagrand} are then of order $\{b^{4}\log n/(nh^{3})\}^{1/2}$ and $b^2\log n/(nh^2)$, whose ratio is $(\log n/nh)^{1/2}$. This tends to zero under Assumption~\ref{as:bw}, because $nh^{3}/(b^4\log n)\to\infty$, hence $nh/\log n\to\infty$. The second term is therefore asymptotically negligible. By Markov's inequality applied to \eqref{eq:talagrand}, we have
\begin{equation}
\sup_{t \in \mathcal{T}} \left| (\mathrm{P}_n - \mathrm{P})\xi(t; \cdot) \right| = O_p\left( \sqrt{ \frac{b^4 \log n}{n h^3} } \right),
\label{eq:final_stochastic_rate}
\end{equation}
which is the second term given in $\rho_n$ in Theorem~\ref{thm:main}.

\noindent{\bf The smoothing bias}
The smoothing bias $\mathbb{E}\widehat{S}_{\mathrm{NA}}(t) - S_T(t)$ can be written as
\begin{equation}
\mathbb{E}\widehat{S}_{\mathrm{NA}}(t) - S_T(t) = - S_T(t) \int_{\ell_0}^{\log t} \frac{\mathbb{E}\widehat{g}_1(u) - g_1(u)}{H_Y(u)} du + S_T(t) \int_{\ell_0}^{\log t} \frac{\mathbb{E}\widehat{H}_Y(u) - H_Y(u)}{H_Y^2(u)} g_1(u) du + \mathbb{E} R_n(t).
\label{eq:bias_expansion}
\end{equation}
We will show that the proposed estimator with the deconvolution kernel $\mathcal{K}_h(\cdot)$ removes the bias due to Laplace perturbation in expectation, reducing the expectation to standard kernel smoothing over the true underlying distribution. First, we evaluate $\mathbb{E}\widehat{g}_1(u)$. Recall that $\widehat{g}_1(u) = \frac{1}{n} \sum_{i=1}^n \widehat{\delta}_i \mathcal{K}_h(u - Y_i^*)$, $\mathbb{E}\widehat{g}_1(u)$ thus is
\begin{align*}
\mathbb{E}\{\widehat{g}_1(u)\} &= \mathbb{E}\left\{ \mathcal{K}_h(u - Y^*_i) \widehat{\delta}_i \right\} \\
&= \mathbb{E}\left[ \mathbb{E}\left\{ \mathcal{K}_h(u - Y_i - e_i) \widehat{\delta}_i \,\middle|\, Y_i, \delta_i \right\} \right].
\end{align*}
Because the privatization mechanisms for $Y_i$ and $\delta_i$ are independent and $Y_i$ and $e_i$ are independent given Assumption~\ref{as:ancillary}, and $\E[\widehat{\delta}_{i}\mid\delta_{i}]=\delta_{i}$, we have
\begin{align*}
\mathbb{E}\left\{ \mathcal{K}_h(u - Y_i - e_i) \widehat{\delta}_i \,\middle|\, Y_i, \delta_i \right\} &= \mathbb{E}\{\widehat{\delta}_i \mid Y_i, \delta_i\} \cdot \mathbb{E}\left\{ \mathcal{K}_h(u - Y_i - e_i)|Y_i \right\} \\
&= \delta_i \cdot \mathbb{E}\left\{ \mathcal{K}_h(u - Y_i - e_i)|Y_i \right\}=\delta_i K_h(u-Y_i).
\end{align*}
Hence, $\mathbb{E}\{\widehat{g}_1(u)\}$ is reduced to a convolution with non-private data $\{(Y_i,\delta_i),i=1,\cdots,n\}$:
\begin{equation}
  \mathbb{E}\{\widehat{g}_1(u)\} = \mathbb{E}\left\{ \delta_i {K}_h(u - Y_i) \right\} = \int_{-\infty}^{\infty} K_h(u - v) g_1(v) dv = (K_h * g_1)(u).
  \label{eq:g1-conv}
\end{equation}
Under Assumptions ~\ref{as:smooth} and ~\ref{as:kernel}, we derive an $r$-th order Taylor expansion to $g_1(u)$ in \eqref{eq:g1-conv}, which is, 
\begin{equation}
  \mathbb{E}\widehat{g}_1(u) - g_1(u) = \frac{(-1)^r h^r}{r!} \mu_r(K) g_1^{(r)}(u) + o(h^r).
  \label{eq:g1_bias_derived}
\end{equation}
Similarly, for the estimator $\widehat{H}_Y(u) = \int_u^\infty \widehat{g}_Y(v) dv$ where $\widehat{g}_Y(v) = \frac{1}{n} \sum_{i=1}^n \mathcal{K}_h(v - Y_i^*)$, we have $\mathbb{E}\widehat{g}_Y(v) = (K_h * g_Y)(v)$. Integrating and expanding $g_Y$ to order $r$ yields:
\begin{equation}
  \mathbb{E}\widehat{H}_Y(u) - H_Y(u) = \int_u^\infty \left\{ \mathbb{E}\widehat{g}_Y(v) - g_Y(v) \right\} dv = \frac{(-1)^r h^r}{r!} \mu_r(K) H_Y^{(r)}(u) + o(h^r),
  \label{eq:HY_bias_derived}
\end{equation}
where $H_Y^{(r)}(u) \triangleq \frac{d^{r-1}}{du^{r-1}}[-g_Y(u)]$.

Substituting \eqref{eq:g1_bias_derived} and \eqref{eq:HY_bias_derived} into \eqref{eq:bias_expansion} yields the uniform bias bound:
\begin{equation}
  \sup_{t \in \mathcal{T}} \left| \mathbb{E}\widehat{S}_{\mathrm{NA}}(t) - S_T(t) \right| \le \sup_{t\in\mathcal{T}}|h^r B_S(t)| + o(h^r) + \sup_{t\in\mathcal{T}}|\mathbb{E}R_n(t)| = O(h^r)+\sup_{t\in\mathcal{T}}|\mathbb{E}R_n(t)|,
  \label{eq:step2_bias_final}
\end{equation}
where $B_S(t) \triangleq \frac{(-1)^{r+1} \mu_r(K)}{r!} S_T(t) \left[ \int_{\ell_0}^{\log t} \frac{g_1^{(r)}(u)}{H_Y(u)} du - \int_{\ell_0}^{\log t} \frac{H_Y^{(r)}(u) g_1(u)}{H_Y^2(u)} du \right]$ is bounded on $\mathcal{T}$ by Assumptions~\ref{as:trunc} and \ref{as:smooth}.

Combining the uniform bounds for stochastic error in \eqref{eq:final_stochastic_rate} and for the smoothing bias in \eqref{eq:step2_bias_final}, we have shown that
\[
 \sup_{t \in \mathcal{T}} \left|\widehat{S}_{\mathrm{NA}}(t) - S_T(t)\right|=O_p\!\left( \sqrt{ \frac{b^4 \log n}{n h^3} } +h^r\right)+\sup_{t\in\mathcal{T}}|R_n(t)|.
\]
In the following, we will show that $R_n(t)$ is negligible comparing with the stochastic error and the smoothing bias. 

\noindent{\bf The asymptotically negligible remainder} The remainder $R_n(t)$ in the functional Taylor expansion \eqref{eq:full_hadamard_proof} satisfies
\begin{equation}
\sup_{t \in \mathcal{T}} |R_n(t)| \le C^*_R \left( \|\widehat{H}_1 - H_1\|^2_\infty + \|\widehat{H}_Y - H_Y\|^2_\infty \right).
\label{eq:remainder_bound_def}
\end{equation}
The classes associated with $\widehat{H}_1$ and $\widehat{H}_Y$, namely $\{y^*\mapsto\widehat{\delta}\{\overline{\mathcal{K}}_h(u-y^*)-\overline{\mathcal{K}}_h(\ell_0-y^*)\}\}$ and $\{y^*\mapsto\overline{\mathcal{K}}_h(\infty)-\overline{\mathcal{K}}_h(u-y^*)\}$ indexed by $u\in[\ell_0,\log\tau']$, share the envelope \eqref{eq:envelope_laplace}, the covering bound \eqref{eq:covernumber} and the weak variance bound \eqref{eq:variance_laplace}. By triangle inequality, we therefore decompose the total uniform error of the component estimators into stochastic and bias parts:
\[
\|\widehat{H}_1 - H_1\|_\infty \le \|\widehat{H}_1 - \mathbb{E}\widehat{H}_1\|_\infty + \|\mathbb{E}\widehat{H}_1 - H_1\|_\infty = O_p\left( \sqrt{\frac{b^4 \log n}{n h^3}} \right) + O(h^r),
\]
and similarly for $\|\widehat{H}_Y - H_Y\|_\infty$. Substituting these uniform rates into \eqref{eq:remainder_bound_def} yields
\begin{equation}
\sup_{t \in \mathcal{T}} |R_n(t)| = O_p\left( h^{2r} + \frac{b^4 \log n}{n h^3} \right), \qquad \sup_{t\in\mathcal{T}}|\mathbb{E}R_n(t)| = O\left( h^{2r} + \frac{b^4 \log n}{n h^3} \right).
\label{eq:remainder_rate}
\end{equation}
Under Assumption~\ref{as:bw}, $n h^3 /(b^4\log n) \to \infty$ and $h \to 0$, so $h^{2r} = o(h^r)$ and $\frac{b^4 \log n}{n h^3} = o\left(\sqrt{\frac{b^4 \log n}{n h^3}}\right)$. Thus, $R_n(t)$ and its centered version $\widetilde{R}_n(t) = R_n(t) - \mathbb{E}R_n(t)$ are asymptotically dominated by the stochastic error and the smoothing bias terms.

Combining the stochastic error rate, the smoothing bias rate, and the remainder rate \eqref{eq:remainder_rate} via triangle inequality, we have the uniform convergence rate,
\begin{align}
\sup_{t \in \mathcal{T}} \left| \widehat{S}_{\mathrm{NA}}(t) - S_T(t) \right| &\le \sup_{t \in \mathcal{T}} \left| \widehat{S}_{\mathrm{NA}}(t) - \mathbb{E}\widehat{S}_{\mathrm{NA}}(t) \right| + \sup_{t \in \mathcal{T}} \left| \mathbb{E}\widehat{S}_{\mathrm{NA}}(t) - S_T(t) \right| \nonumber \\
&= O_p\left( \sqrt{ \frac{b^4 \log n}{n h^3} } \right) + O(h^r) + O_p\left( h^{2r} + \frac{b^4 \log n}{n h^3} \right) \nonumber \\
&= O_p\left( h^r + \sqrt{ \frac{b^4 \log n}{n h^3} } \right).
\label{eq:combined_rate}
\end{align}
Balancing the deterministic smoothing bias $O(h^r)$ against the stochastic term $O_p\left(\sqrt{\frac{b^4 \log n}{n h^3}}\right)$ yields the bandwidth to achieve the uniform convergence rate:
\begin{equation}
h \asymp \left( \frac{b^4 \log n}{n} \right)^{\frac{1}{2r + 3}}.
\label{eq:optimal_bandwidth}
\end{equation}
Substituting $h$ back into \eqref{eq:combined_rate} establishes the overall uniform convergence rate for Theorem~\ref{thm:main}:
\begin{equation}
\sup_{t \in \mathcal{T}} \left| \widehat{S}_{\mathrm{NA}}(t) - S_T(t) \right| = O_p\left( \left( \frac{b^4 \log n}{n} \right)^{\frac{r}{2r + 3}} \right).
\label{eq:final_optimal_rate}
\end{equation}
When the privacy budget is fixed, the scale parameter $b$ in the Laplace distribution for the perturbations remains constant as $n \to \infty$. Therefore, $\frac{b^4 \log n}{n} = O\left(\frac{\log n}{n}\right) \longrightarrow 0$ as $n \to \infty$. Hence, $\left( \frac{b^4 \log n}{n} \right)^{\frac{r}{2r + 3}} \to 0$, which implies the uniform consistency of the proposed estimator, that is, $\sup_{t \in \mathcal{T}} \left| \widehat{S}_{\mathrm{NA}}(t) - S_T(t) \right| = o_p(1)$.
\end{proof}
\section*{B: Proof of Theorem 2}
\begin{proof}
From \eqref{eq: stochastic2}, we have
$$\widehat{S}_{\text{NA}}(t) - \mathbb{E}[\widehat{S}_{\text{NA}}(t)] = \frac{1}{n} \sum_{i=1}^n \left\{ \xi(t;\mathcal{O}^*_i) - \mathbb{E}[\xi(t;\mathcal{O}^*_i)] \right\} + \tilde{R}_n(t)$$ By defining $\varsigma_n^2(t)$ as the variance of $\xi(t;\mathcal{O}^*_i)$, we have
$$\sqrt{n} \varsigma_n^{-1}(t) \left\{ \widehat{S}_{\text{NA}}(t) - \mathbb{E}\widehat{S}_{\text{NA}}(t) \right\} = \frac{1}{\sqrt{n}\varsigma_n(t)} \sum_{i=1}^n\left\{\xi(t;\mathcal{O}^*_i) - \mathbb{E}[\xi(t;\mathcal{O}^*_i)] \right\} + \frac{\sqrt{n}}{\varsigma_n(t)} \tilde{R}_n(t).$$
From \eqref{eq:remainder_rate}, we have shown that $\tilde{R}_n(t)$ is asymptotically negligible against $(\mathrm{P}_n-\mathrm{P})\xi(t;\mathcal{O}^*)$. Now we will look at whether $\frac{\sqrt{n}}{\varsigma_n(t)} \widetilde{R}_n(t)$ is asymptotically negligible. We present this scaled reminder term as
\[\frac{\sqrt{n}}{\varsigma_n(t)} \widetilde{R}_n(t) = O_p\left( \frac{\sqrt{n h^3}}{b^2} \left[ h^{2r} + \frac{b^4 \log n}{n h^3} \right] \right) = O_p\left( \frac{\sqrt{n h^{4r+3}}}{b^2} + \frac{b^2 \log n}{\sqrt{n h^3}} \right).\]
Under Assumption \ref{as:bw} $\frac{n h^3}{b^4 \log n} \to \infty$, the second term $\frac{b^2 \log n}{\sqrt{n h^3}} = \sqrt{\frac{b^4 \log^2 n}{n h^3}} = \sqrt{\log n \cdot \frac{b^4 \log n}{n h^3}} \longrightarrow 0$. With the bandwidth bounded in \eqref{eq:optimal_bandwidth},the first term $\sqrt{n h^{4r+3}} = h^r \sqrt{n h^{2r+3}} \to 0$. Therefore, $$\frac{\sqrt{n}}{\varsigma_n(t)} \widetilde{R}_n(t) = o_p(1).$$Further, we have $\Vert{}\xi(t;\mathcal{O}^*_i) - \mathbb{E}[\xi(t;\mathcal{O}^*_i)]\Vert{}_\infty \le C^*_1h^{-2}$ and $\varsigma_n^2(t)\asymp h^{-3}$ from \eqref{eq:envelope_laplace} and \eqref{eq:variance_laplace}. Therefore, by the central limit theorem,
$$\frac{1}{\sqrt{n} \varsigma_n(t)} \sum_{i=1}^n \left( \xi(t; O_i^*) - \mathbb{E}[\xi(t)] \right) \stackrel{d}{\longrightarrow} {N}(0, 1).$$With the asymptotically negligible $\tilde{R}_n(t)$, 
we establish the asymptotic normality for $\hat{S}_{\mathrm{NA}}(t)$,
\[\sqrt{n}\,\varsigma_{n}^{-1}(t)
  \left\{\widehat S_{\mathrm{NA}}(t)-\E\widehat S_{\mathrm{NA}}(t)\right\}
  \stackrel{d}{\longrightarrow} N(0,1).\]
Next, we show the consistency of variance estimator via the multiplier bootstrap procedure. Because $\E\omega=1$, the bootstrap quantities are centered at the estimators derived from the sample, for example, $\E\{\widehat g^{\omega}_{1}(u)\mid\mathcal{O}^*\}=\widehat g_{1}(u)$ and
$\E\{\widehat H^{\omega}_{Y}(u)\mid\mathcal{O}^*\}=\widehat H_{Y}(u)$ where $\widehat g_{1}(u)$ and $\widehat H_{Y}(u)$ are estimators in \eqref{eq:ghat} and \eqref{eq:Hhat}. Consequently, by a first-order functional Taylor expansion of $\widehat{S}_{\mathrm{NA}}^\omega(t)$ around the estimators $(\widehat{g}_1, \widehat{H}_Y)$, the reweighted estimator is :$$\widehat{S}_{\mathrm{NA}}^\omega(t) - \widehat{S}_{\mathrm{NA}}(t) = \frac{1}{n} \sum_{i=1}^n (\omega_i - 1) \widehat\xi(t;\mathcal{O}^*_i) + R_{n}^\omega(t),$$ where
$\widehat{\xi}(t;\mathcal{O}^*_i) = -\widehat{S}_{\mathrm{NA}}(t) \int_{\log t_0}^{\log t} \left[ \frac{\widehat{\delta}_i \mathcal{K}_h(u - Y_i^*)}{\widehat{H}_Y(u)} - \frac{\widehat{g}_1(u) \int_u^\infty \mathcal{K}_h(v - Y_i^*) \, dv}{\widehat{H}_Y^2(u)} \right] du$ and the remainder term $R_{n}^\omega(t)$ is asymptotically negligible. Because the weights $\omega_1, \dots, \omega_n$ are i.i.d. with $\mathrm{Var}(\omega_i \mid \mathcal{O}^*) = \mathrm{Var}(\omega_i) = 1$, the conditional variance of $\widehat{S}_{\mathrm{NA}}^\omega(t)$ given $\mathcal{O}^*$ simplifies to: $$\widehat{\varsigma}_n^2(t) = n \mathrm{Var}\left\{ \widehat{S}_{\mathrm{NA}}^\omega(t) \;\middle\vert{}\; \mathcal{O}^* \right\} = n \mathrm{Var}\left\{ \frac{1}{n} \sum_{i=1}^n (\omega_i - 1) \widehat\xi(t;\mathcal{O}^*_i) + R_{n}^\omega(t) \;\middle\vert{}\; \mathcal{O}^* \right\} = \frac{1}{n} \sum_{i=1}^n \widehat\xi^2(t;\mathcal{O}^*_i) + o_p(1).$$
To establish $\widehat\varsigma_n^2(t)/\varsigma_n^2(t) \stackrel{p}{\longrightarrow} 1$, we decompose:
\begin{equation*}
  \left| \frac{\widehat\varsigma_n^2(t)}{\varsigma_n^2(t)} - 1 \right| 
  \le \frac{1}{\varsigma_n^2(t)} \left| \frac{1}{n}\sum_{i=1}^n \widehat{\xi}^2(t;\mathcal{O}^*_i) - \frac{1}{n}\sum_{i=1}^n \xi^2(t;\mathcal{O}^*_i) \right| 
  + \left| \frac{n^{-1}\sum_{i=1}^n \xi^2(t;\mathcal{O}^*_i)}{\varsigma_n^2(t)} - 1 \right|.
\end{equation*}
By a first-order Taylor expansion of integrand in $\widehat{\xi}(t;\mathcal{O}^*_i)$, we have $\sup_i |\widehat{\xi}(t;\mathcal{O}^*_i) - \xi(t;\mathcal{O}^*_i)| = O_p(\rho_n h^{-\beta})$. Since $\varsigma_n^2(t) \asymp h^{-(2\beta - 1)}$, the first term satisfies $O_p(\rho_n^2 h^{-1}) = o_p(1)$ under Assumption~\ref{as:bw}.

For the second term, taking expectations gives $\E\{\xi^2(t;\mathcal{O}^*_i)\} = \varsigma_n^2(t)$, and the envelope bound $\|\xi(t;\mathcal{O}^*_i)\|_\infty \le C^* h^{-\beta}$ implies $\E\{\xi_1^4(t)\} \le {C^{*}}^2 h^{-2\beta} \varsigma_n^2(t)$. Applying Chebyshev's inequality yields
\begin{equation*}
  \Var\left( \frac{n^{-1} \sum_{i=1}^n \xi^2(t;\mathcal{O}^*_i)}{\varsigma_n^2(t)} \right) 
  \le \frac{\E[\xi^4(t;\mathcal{O}^*_i]}{n \varsigma_n^4(t)} 
  \le \frac{ {C^{*}}^2 h^{-2\beta} \varsigma_n^2(t)}{n \varsigma_n^4(t)} 
  = O\left( \frac{1}{nh} \right) \rightarrow 0.
\end{equation*}
Combining both bounds confirms that $\widehat\varsigma_n^2(t) / \varsigma_n^2(t) \stackrel{p}{\longrightarrow} 1$ as $n \to \infty$.
  \end{proof}

\end{document}